\documentclass[11pt]{article}
\usepackage{include-file}
\begin{document}
\title{Approximating Multicut and the Demand Graph\footnote{Department of Computer Science,
University of Illinois, Urbana, IL 61801. {\tt \{chekuri,vmadan2\}@illinois.edu}. Work on this paper partly supported by NSF grant CCF-1319376.} }
\author{Chandra Chekuri \and Vivek Madan}
\date{\today}

\pagenumbering{gobble}

\maketitle

\begin{abstract}
In the minimum Multicut problem, the input is an edge-weighted supply
graph $G=(V,E)$ and a simple demand graph $H=(V,F)$. Either $G$ and
$H$ are directed (\DMulC) or both are undirected (\UMulC). The goal is
to remove a minimum weight set of edges $E' \subseteq E$ such that for
any edge $(s,t) \in F$, there is no path from $s$ to $t$ in the graph
$G-E'$. \UMulC admits an $O(\log k)$-approximation where $k$ is the
vertex cover size of $H$ while the best known approximation for \DMulC
is $\min\{k, \tilde{O}(n^{11/23})\}$. These approximations are
obtained by proving corresponding results on the multicommodity
flow-cut gap.  In contrast to these results some special cases of
Multicut, such as the well-studied Multiway Cut problem, admit a
constant factor approximation in both undirected and directed graphs.
In this paper, motivated by both concrete instances from applications
and abstract considerations, we consider the role that the structure
of the demand graph $H$ plays in determining the approximability of
Multicut. We obtain several new positive and negative results.

In undirected graphs our main result is a $2$-approximation in
$n^{O(t)}$ time when the demand graph $H$ excludes an induced matching of
size $t$. This gives a constant factor approximation for a specific
demand graph that motivated this work, and is based on a reduction
to uniform metric labeling, and not via the flow-cut gap.

In contrast to the positive result for undirected graphs, we prove
that in directed graphs such approximation algorithms can not exist.
We prove that, assuming the Unique Games Conjecture (UGC), that for a
large class of fixed demand graphs \DMulC cannot be approximated to a
factor better than worst-case flow-cut gap. As a consequence we prove
that for any fixed $k$, assuming UGC, \DMulC with $k$ demand pairs is
hard to approximate to within a factor better than $k$. On the
positive side, we prove an approximation of $k$ when the demand graph
excludes certain graphs as an induced subgraph. This positive 
result generalizes the $2$ approximation for directed Multiway Cut
to a much larger class of demand graphs.

\end{abstract}

\newpage

\pagenumbering{arabic}

\section{Introduction}
The minimum \MulC problem is a generalization of the classical
$s$-$t$ cut problem to multiple pairs. The input to the \MulC
problem is an edge-weighted graph $G=(V,E)$ and $k$ source-sink pairs
$(s_1,t_1),(s_2,t_2),\ldots,(s_k,t_k)$. The goal is to find a minimum
weight subset of edges $E' \subseteq E$ such that
all the given pairs are disconnected in $G-E'$; that is, for 
$1 \le i \le k$, there is no path from $s_i$ to $t_i$ in $G-E'$.
In this paper we consider an equivalent formulation that exposes,
more directly, the structure that the source-sink pairs may have.

The input now consists of an edge-weighted {\em supply} graph
$G=(V,E)$ and a {\em demand} graph $H=(V,F)$. The goal is to find a
minimum weight set of edges $E' \subseteq E$ such that for each edge
$f=(s,t) \in F$, there is no path from $s$ to $t$ in $G-E'$. In other
words the source-sink pairs are encoded in the form of the demand
graph $H$.  Either both $G$ and $H$ are directed in which case we
refer to the problem as \DMulC (directed \MulC) or both are
undirected in which case we refer to the problem as \UMulC (undirected
\MulC).

\MulC in both directed and undirected graphs has been extensively
studied for a variety of reasons. It is a natural cut problem, has
several applications, and strong connections to several other
well-known problems such as sparsest cut and multicommodity flows.
\UMulC and \DMulC are NP-Hard even in very restrictive settings. For
instance \UMulC is NP-Hard even when $H$ has 3 edges; it generalizes
\VC even when $G$ is a tree.  \DMulC is NP-Hard and APX-Hard even in
the special case when $H$ is a cycle of length $2$ which is better
understood as removing a minimum weight set of edges to disconnect $s$
from $t$ {\em and} $t$ from $s$ in a directed graph. Consequently
there has been substantial effort towards developing approximation
algorithm for these problems as well as understanding special
cases. We briefly summarize some of the known results. We use $k$ to
denote the number of edges in the demand graph $H$.  For \UMulC there
is an $O(\log k)$-approximation \cite{GVY96} which improves to an
$O(r)$-approximation if the supply graph $G$ excludes $K_r$ as a minor
(in particular this yields a constant factor approximation in planar
graphs) \cite{AbrahamGGOK14,FT03,KPR93}. In terms of
inapproximability, \UMulC is at least as hard as \VC even in trees and
hence APX-Hard.  Under the Unique Game Conjecture (UGC) it is known to
be super-constant hard \cite{ChawlaKKRS06}.  \DMulC is much
harder. The best known approximation is
$\min\{k,\tilde{O}(n^{11/23})\}$; here $n = |V|$. Note that a
$k$-approximation is trivial. 
\begin{wrapfigure}{r}{.38\textwidth}
  \vspace{-1em}
\begin{boxedminipage}{0.38\linewidth}
\vspace{-0.2in}
\begin{align*}
\min \quad & \sum_{e \in E} w_e x_e & \\
   & \sum_{e \in p} x_e & \ge 1 & \qquad p \in \cP_{st}, st \in F \\
&  x_e & \ge 0 & \qquad e \in E
\end{align*}
  \label{fig:mulc_lp}
  \vspace{-0.2in}
\end{boxedminipage}
  \vspace{-1em}
\end{wrapfigure}
Moreover, \DMulC is hard to approximate to within a factor of
$\Omega(2^{\log^{1-\epsilon} n})$ assuming $NP \neq
ZPP$~\cite{ChuzhoyK09}; evidence is also presented in
\cite{ChuzhoyK09} that it could be hard to approximate to within a
polynomial factor.  We note that all the preceding positive results
for \MulC are based on bounding the integrality gap of a natural LP
relaxation shown in adjacent figure. This is the standard cut
formulation with an variable for each edge and an exponential set of
constraints which admit a polynomial-time separation oracle; one can
also write a compact polynomial-time formulation. The dual is a
maximum multicommodity flow LP. We henceforth refer to the integrality
gap of this LP as the {\em flow-cut gap}.  Most multicut approximation
algorithms are based on proving bounds on the flow-cut gap.

\medskip
\noindent
{\bf The role of the demand graph:} Our preceding discussion has focused
on the approximability of \MulC when $H$ is arbitrary with some improved
results when $G$ is restricted. However, we are interested here in the
setting where $H$ is restricted and $G$ is arbitrary. Before we describe
a concrete application that motivated us, we mention the well-known
\MCut problem in undirected graphs (\UMCut) and directed graphs (\DMCut). Here
$H$ is the complete graph on a set of $k$ terminals. This problem has
been extensively studied over the years and a constant factor approximation
is known. \UMCut admits a $1.29$ approximation \cite{SharmaV14} and
\DMCut admits a $2$-approximation \cite{NaorZ97,ChekuriM16}.
In a recent work, motivated by connections to the problem of
understanding the information capacity of networks with delay constraints \cite{ChekuriKKV15}, the following special
case of \MulC was introduced. It was referred to as the 
{\sc Triangle-Cast} problem. 
\begin{wrapfigure}{r}{0.25\linewidth}
\vspace{-1em}
\begin{minipage}{1\linewidth}
  \begin{center}
  \includegraphics[width=1.5in]{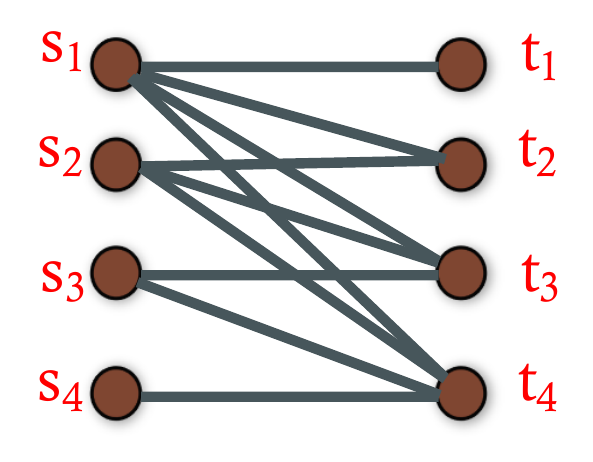}    
  \end{center}
%\caption{Triangle Cast Pattern: $(s_i,t_j)$ is an edge iff $i\leq j$.}
%\label{fig:tri_cast_pattern} 
\end{minipage} 
\vspace{-2em}
\end{wrapfigure}
The demand graph $H$ is a bipartite graph with $k$ terminals
$s_1,\ldots,s_k$ on one side and $k$ terminals $t_1,\ldots,t_k$ on the
other side: $(s_i,t_j)$ is an edge in $H$ iff $i \le j$. See figure
for an example with $k=4$. It was shown in \cite{ChekuriKKV15} that
the flow-cut gap for this special case of \MulC in directed graphs
gives an upper bound on the capacity advantage of network
coding. Further, it was established that the flow-cut gap is $O(\log
k)$ even in directed graphs which is in contrast to the general
setting where the gap can be as large as $k$. The following natural
questions arose from this application.

\begin{question}
  What is the approximability of \TriCast in directed and undirected graphs?
  Is the flow-cut gap $O(1)$ in undirected graphs and even in directed graphs?
\end{question}

Answering the preceding question has not been easy. In fact we do not
yet know whether the flow-cut gap is $O(1)$ even in undirected graphs.
However, in this paper we give two answers. First, we give a
$2$-approximation for \TriCast in undirected graphs via a different LP
relaxation.  Second, we show that under UGC, for any fixed constant
$k$, the hardness of approximation for \TriCast in directed graphs
co-incides with with the flow-cut gap.  At this moment we only know
that the flow-cut gap is $O(\log k)$ and at least some fixed constant
$c > 1$. We mention that \TriCast in directed graphs is approximation
equivalent to another problem that has recently been considered in
with a different motivation called \LinCut \cite{ErbacherJTT14}; here
the demand graph $H$ consists of $k$ terminals $s_1,\ldots,s_k$ and
there is a directed edge $(s_i,s_j)$ for all $i < j$.

Our results for \TriCast are special cases of more general results
that examine the role that the demand graph $H$ plays in the
approximability of \MulC. What structural aspects of $H$ allow
for better bounds than the worst-case results?  For instance, do
the constant factor approximation algorithms for \MCut be understood
in a more general setting?
\iffalse
We were also inspired by the following
related question. Can we understand the fact that \MCut has
a constant factor approximation in a more general setting? 
In a recent paper on a simpler $2$-approximation for
\DMCut \cite{ChekuriM16} we raised the question of whether
the case of two terminals is hard to approximate to within a
factor of $2$. This was motivated by the fact that the flow-cut
gap even for this case is $2$. In this paper we answer this
question in the positive under UGC. 
\fi
Some previous work has also examined the role that demand graph plays in
\MulC. Two examples are the original paper of Garg, Vazirani
and Yannakakis \cite{GVY96} who showed that one can obtain an
$O(\log h)$-approximation for \UMulC 
where $h$ is the vertex cover size of the demand graph. This was generalized
by Steurer and Vishnoi \cite{SteurerV09} who showed that
$h$ can be chosen to be $\min_S \max_T |S \cap T|$ where 
$S$ is a vertex cover in $H$ and $T$ is an independent set in $H$.
Note that both these results are based on the flow-cut gap and
yield only an $O(\log k)$ upper bound for \TriCast.

\iffalse
\begin{figure*}[t]
  \centering
\begin{minipage}{0.45\linewidth}
\vspace{-0.2in}
\begin{align*}
\min \quad & \sum_{e \in E} w_e x_e & \\
   & \sum_{e \in p} x_e & \ge 1 & \qquad p \in \cP_{st}, st \in F \\
&  x_e & \ge 0 & \qquad e \in E
\end{align*}
  \caption{LP Relaxation for \UMulC and \DMulC. $\cP_{st}$ is the
  set of all $s$-$t$ paths in $G$. The dual LP is 
a maximum multicommodity flow LP.}
  \label{fig:mulc_lp}
\end{minipage}
\begin{minipage}{0.45\linewidth}
  \begin{center}
  \includegraphics[width=1.5in]{triangle-cast}    
  \end{center}
\caption{Triangle Cast Pattern: $(s_i,t_j)$ is an edge iff $i\leq j$.}
\label{fig:tri_cast_pattern} 
\end{minipage}
\end{figure*}
\fi

We now describe our results for both \UMulC and \DMulC which give
yield as corollaries the result that we already mentioned and several
other ones.
\subsection{Our Results}
\label{subsec:results}

\iffalse
\vspace{2mm}
{\em There exists a number $t$ such that the demand graph $H$ does not contain an induced subgraph isomorphic to a matching of size $t$}
\vspace{2mm}
\fi

We first discuss our result for \UMulC. We obtain a $2$-approximation for
a class of demand graphs. This class is inspired by the observation
that the \TriCast demand graph does not contain a matching with two edges
as an {\em induced subgraph}\footnote{$G'$ is an induced subgraph 
of $G=(V,E)$ if $G'=G[V']$ for some $V' \subset V$.}. More generally, a graph is
said to be $tK_2$-free for an integer $t > 1$ if it does not contain
a matching of size $t$ as an induced subgraph. 

\begin{restatable}{thm}{tk}\label{thm:tk2_approximation}
  There is a $2$-approximation algorithm with running time
  $\text{poly}(n, k^{O(t)})$ on instances of \UMulC with supply graph $G$ and 
$tK_2$-free demand graph $H$. Here $n = V(G),k = V(H)$.
\end{restatable}

Since \TriCast instances are $2K_2$-free we obtain the following corollary.
\begin{corollary}
\TriCast admits a polynomial-time $2$-approximation in undirected graphs.
\end{corollary}

We note that the preceding approximation is not based on the
natural LP relaxation. It relies on a different relaxation via
a reduction to uniform metric labeling \cite{KleinbergT02}.

\medskip We now turn our attention to \DMulC. As we mentioned the best
approximation in general graphs is $\min\{k, \tilde{O}(n^{11/23})\}$.
It is also known that the flow-cut gap is lower bounded by $k$
\cite{SaksSZ04} for $k = O(\log n)$ and also by
$\tilde{\Omega}(n^{1/7})$ when $k$ can be polynomial in $n$
\cite{ChuzhoyK09}. It is natural to ask about the hardness of the
problem when $k$ is a fixed constant. In particular what is the
relationship between the flow-cut gap and hardness?  To formalize
this, for a {\em fixed} demand graph $H$ we define \DMulCH as the
special case of \DMulC where $G$ is arbitrary but the demand graph is
constrained to be $H$. To be formal we need to define $H$ as a
``pattern'' since even for a fixed supply graph $G$ we need to specify
the nodes of $G$ to which the nodes of $H$ are mapped. However we
avoid further notation since it is relatively easy to understand what
\DMulCH means.  We define $\alpha_H$ to be the worst-case flow-cut
gap over all instances with demand graph $H$. We conjecture the
following general result.

\begin{conjecture}
  For {\em any fixed} demand graph $H$ and any fixed $\eps > 0$, unless $P=NP$,
  there is no polynomial-time $(\alpha_H-\eps)$- approximation for \DMulCH.
\end{conjecture}

In this paper we prove weaker forms of the conjecture,
captured in the following two theorems:

\begin{restatable}{thm}{hardness}\label{thm:hardness}
  Assuming UGC, for any fixed directed {\em bipartite} graph $H$, and
  for any fixed $\eps > 0$ there is no polynomial-time $(\alpha_H-\eps)$
  approximation for \DMulCH.
\end{restatable}

\iffalse
\begin{theorem}
  Assuming UGC, for any fixed directed {\em bipartite} graph $H$, and
  for any fixed $\eps > 0$ there is no polynomial-time $(\alpha_H-\eps)$
  approximation for \DMulCH.
\end{theorem}
\fi

If $H$ is not bipartite we obtain a slightly weaker theorem.

\begin{theorem}
  \label{thm:hardness-nonbipartite}
  Assuming UGC, for any fixed directed graph $H$ on $k$ vertices and
  for any fixed $\eps > 0$, there is no polynomial-time
  $\frac{\alpha_H}{2\lceil \log k \rceil}-\eps$ approximation for  
  \DMulCH.
\end{theorem}

Via known flow-cut gap results  \cite{SaksSZ04} 
and some standard reductions we obtain the following corollary.
\vspace{-4mm}
\begin{corollary}
  \label{cor:hardness}
  Assuming UGC,
  \begin{compactitem}
  \item For any fixed $k$, if $H$ is a collection of $k$ disjoint directed 
    edges then \DMulCH is hard to approximate within a factor of $k-\eps$.
    
  \item Separating $s$ from $t$ \emph{and} $t$ from $s$ in a directed graph
    (\DMCut with $2$ terminals) is hard to approximate within a factor of $2-\eps$.
  \item For any fixed $k$, \TriCast's approximability coincides with 
the flow-cut gap.
  \end{compactitem}
\end{corollary}
\vspace{-2mm}
\iffalse
We now state several important corollaries of Theorem~\ref{thm:main-hardness}.
Using the flow-cut gap established in \cite{SaksSZ04}, 
we obtain the following corollary.

\begin{corollary}
  \label{cor:k-matching}
  Under UGC, if $H$ is a graph with $k$ disjoint edges then \DMulC($H$) is hard
  to approximate to within a factor of $k-\eps$ for any fixed $\eps > 0$.
\end{corollary}

Note that a factor of $k$ approximation is trivial.

Via a standard reduction we have the following answering a question
raised in \cite{ChekuriM16}.
\begin{corollary}
  Under UGC, \DMCut with $2$ terminals is hard to approximate within a
  factor of $2-\eps$ for any fixed $\eps > 0$.
\end{corollary}

Finally, although we do not know yet the precise answer to the flow-cut gap
for \TriCast, we have the following.

\begin{corollary}
  Under UGC, the hardness of \TriCast for any fixed value of $k$
  is equal to the flow-cut gap. 
\end{corollary}
\fi

\medskip Our last result is on upper bounds for \DMulC.  Can we
improve the known approximation bounds based on the structure of the
demand graph? Corollary~\ref{cor:hardness} shows that if $H$ contains,
a matching of size $k$ as an induced subgraph then we cannot obtain a
better than $k$ approximation. The following question arises naturally.

\begin{question}
  Let $H$ be a fixed demand graph such that it does not contain 
  a matching of size $k$ as an induced subgraph. Is there a 
  $k$-approximation for \DMulCH?
\end{question}
\vspace{-1mm}
A positive answer to above question would imply a $2$-approximation
for \TriCast in directed graphs.  Since we are currently unable to
improve the $O(\log k)$-approximation for \TriCast we consider a
relaxed version of the preceding question and give a positive answer.
We say that a directed demand graph $H=(V,F)$ contains
an induced $k$-matching-extension if there are two subsets of $V$,
$S = \{s_1,s_2,\ldots,s_k\}$ and $T=\{t_1,t_2,\ldots,t_k\}$ such that
the induced graph on $S \cup T$ satisfies the following properties:
(i) for $1 \le i \le k$, $(s_i,t_i) \in F$ and (ii) for $i > j$,
$(s_i,t_j) \not \in F$. Not that $s_1,s_2,\ldots,s_k$ are distinct
since $S$ is a set and similarly $t_1,\ldots,t_k$ are distinct but
some $s_i$ may be the same as a $t_j$ for $i \neq j$.
We give two examples to illustrate the utility of considering this
special case of instances.

Consider \DMCut which corresponds to the demand graph $H$ being a
complete directed graph. It can be verified that $H$ does {\em not}
contain an induced $3$-matching-extension. Now consider $H=(V,F)$
being a complete graph on an even number of nodes and remove edges
from $H$ corresponding to a perfect matching $M$ on $V$ (if $uv \in M$
we remove $(u,v)$ and $(v,u)$ from $F$). Let $H'$ be the resulting
demand graph. We claim that $H'$ does not contain an induced
$4$-matching-extension to \TriCast. What is the approximability of
\DMulCH'? It is fair to say that previous work would not have
found it easy to answer this question since $H'$ does not appear to
have the same nice structure that the complete graph has. The theorem
below shows that one can obtain a $3$ approximation for
\DMulCH'.

\iffalse
Recall that for \DMCut we
have a $2$-approximation and for \TriCast we have an $O(\log
k)$-approximation. We build on our recent work on \DMCut to show
that we can obtain a non-trivial bound on the flow-cut gap and
a corresponding approximation bound, for significant generalizations
of \DMCut. Interestingly the characterization involves \TriCast pattern.
The formal statement of our result is given below.
\fi

\iffalse
In a similar spirit as undirected graphs, we also prove that if the demand graph does not contain an induced graph which lies between $k$ parallel edges and the triangle cast pattern(figure~\ref{fig:tri_cast_pattern}) with parameter $k$, then there exists a factor $k-1$ approximation. 

{\bf Chandra:} define the property of the demand graph separately
in a clean fashion instead of introducing it inside the theorem.

\begin{restatable}{thm}{dirapproximation}\label{thm:dir_approximation}
  Consider \DMulC instance $(G,H)$ such that for some integer $h$,
  there does not exists any induced subgraph of $H=(V,F)$ on vertices
  $s_1,\dots,s_h,t_1,\dots,t_h$ such that the following properties are
  satisfied: $(i)$ $\forall i \in [1,h], s_it_i\in F$, $(ii) \forall
  i>j \in [1,h], s_i t_j \not\in F$. $s_i$'s are all distinct and
  $t_i$'s are all distinct but some $s_i$ might be same as some $t_j$.

  If demand graph $H$ satisfies this property, then the flow-cut gap is
  at most $k-1$ and there is a polynomial-time rounding algorithm that
  achieves this upper bound.
\end{restatable}
\fi

\begin{restatable}{thm}{dirapproximation}\label{thm:dir_approximation}
  Consider \DMulCH where $H$ does not contain an induced
  $k$-matching-extension.  The the flow-cut gap is at most $k-1$
  and there is a polynomial-time rounding algorithm that achieves this
  upper bound.
\end{restatable}

The rounding scheme that proves the preceding theorem is built upon
our recent insight for \DMCut \cite{ChekuriM16}.  Interestingly the
rounding scheme is itself oblivious to the demand graph $H$. It either
provably obtains a $(k-1)$ approximation via the LP solution or
provides a certificate that $H$ contains an induced
$k$-matching-extension. 

\medskip
\noindent {\bf Techniques:} At a high-level our main insights are
based on a \emph{labeling} view for \MulC instead of using the
standard LP based on distances. For undirected graphs we show that
this yields Theorem~\ref{thm:tk2_approximation}. In directed graphs we
show that a labeling based LP is equivalent to the standard
LP which is starkly in contrast to the undirected graph setting.
The labeling LP allows us to relate the hardness of
\DMulCH to the hardness of constraint satisfaction problems via
a standard labeling LP for CSPs called \BasicLP. We crucially rely
on a general hardness result for \MinCSP due to Ene, Vondrak and
Wu \cite{EneVW13} which generalized prior work of Manokaran \etal
\cite{ManokaranNRS08}. Finally, Theorem~\ref{thm:dir_approximation}
is builds upon our recent insights into rounding for
\DMCut \cite{ChekuriM16}.

\iffalse
\medskip
\noindent
{\bf Other Related Work:} \colorcomment{Add a discussion of Vishnoi-Steurer paper.
Also result of Cheriyan-Karloff-Rabani on directed multicut. 
FPT results? Paper of Marx-Wollan on demand graphs that make 
max disjoing paths fpt.}
\fi

\medskip
\noindent {\bf Organization:} Section~\ref{sec:tK2_MulC} describes the
factor $2$-approximation for
\tK2MulC. Section~\ref{sec:DMulC_hardness} describes the proof of
hardness of approximation for
\DMulCH. Section~\ref{sec:DMulC_approximation} described the $k-1$
approximation for \DMulCH when $H$ does not contain an induced
$k$-matching extension. Due to space constraints many of the proofs,
including that of Theorem~\ref{thm:hardness-nonbipartite}, 
are provided in the appendix. 
\vspace{-4mm}

\section{Approximating \UMulC with $tK_2$-free demand graphs}\label{sec:tK2_MulC}
In this section we obtain $2$-approximation for $tK_2$-free demand
graphs and prove Theorem~\ref{thm:tk2_approximation}.
%\tk*
Before we prove the theorem, we consider the \UMulC problem
where demand graph has some fixed size $k$. Given supply graph
$G=(V,E)$ let $S = \{s_1,\dots,s_k\} \subset V$ be the {\em terminals}
participating in the demand edges specified by $H$. A feasible
solution $E' \subset E_G$ of the \UMulC instance will induce a
partition over $S$ such that if $s_is_j$ is an edge in the demand
graph $H$, then $s_i$ and $s_j$ belong to different components in
$G-E'$. Note that two terminals that are not connected by a demand
edge may be in the same connected compoent of $G-E'$.  If $k$ is a
fixed constant we can ``guess'' the partition of the terminals induced
by an optimum solution. With the guess in place it is easy to see that
the problem reduces to an instance of \UMCut which admits a constant
factor approximation.
Thus, one can obtain a constant factor approximation for \UMulC in
$2^{O(k \log k)} \text{poly}(n)$ time by trying all possible partitions of
the terminals.

To prove Theorem~\ref{thm:tk2_approximation}, we use this idea of
enumerating feasible partitions. However, $H$ is not necessary of
fixed size, and enumerating all possible partitions of the terminals
is not feasible. Instead, we make use of the following theorem which
bounds the number of {\em maximal} independent sets in a $tK_2$-free
graph.

\begin{theorem}(Balas and Yu \cite{BalasY89})
  \label{thm:tK2_independent_sets}
  Any $s$-vertex $tK_2$-free graph has at most $s^{O(t)}$ maximal
  independent sets and these can be found in $s^{O(t)}$ time.
\end{theorem}

We prove Theorem~\ref{thm:tk2_approximation} by using the preceding
theorem and reducing the \UMulC problem to the \UML problem.
We now describe the general \ML problem.

\noindent
\ML: The input consists of an undirected edge-weighted graph $G =
(V,E)$, a set of labels $L=\{1,\dots,h\}$ and a metric $d(i,j) , i,j
\in L$ defined over the labels. In addition for each vertex $u \in V$
and label $i \in L$ there is a non-negative assignment cost $c(u,i)$.
Given an assignment $f: V \rightarrow L$ of vertices to labels we
define its cost as $\sum_{u \in V} c(u, f(u)) + \sum_{uv \in E} w(uv)
d(f(u),f(v))$. The goal is to find an assignment of minimum cost.
The special case when the metric is uniform, that is
$d(i,j) = 1$ for $i \neq j$, is refered to as \UML.

\begin{theorem}\label{thm:uml_approximation}(Kleinberg and Tardos~\cite{KleinbergT02})
  There is a $2$-approximation for \UML.
\end{theorem}

\begin{proofof}{Theorem ~\ref{thm:tk2_approximation}}
  Let the demand graph $H$ of the \UMulC instance be
  $tK_2$-free. Using Theorem~\ref{thm:tK2_independent_sets}, we can
  find all maximal independent sets in $H$. Let these independent sets
  be $I_1,\dots,I_r$ where $r \leq |V_H|^{O(t)}$. 
  Note that the indepdendent sets are
  considered only in the demand graph. 

  Consider the following instance of \UML:
  The supply graph $G=(V,E)$ of the \UMulC instance is 
  the input graph to the \UML instance. The label set
  $L = \{1,2,\ldots,r\}$, one for each maximal independent set in $H$.
  For each $u \in V(H)$  let $c(u,i) = 0$ if $u \in I_i$ and
  $c(u,i) = \infty$ otherwise. For a vertex $u \in V$ where $u$ is
  not a terminal we have $c(u,i) = 0$ for all $i$.
  
  We claim that the preceding reduction is approximation preserving.
  Assuming the claim, we can obtain the desired $2$-approximation by
  solving the \UML instance using
  Theorem~\ref{thm:uml_approximation}. The size of the \UML instance
  that is generated from the given \UMulC instance is $\text{poly}(n,
  |V_H|^{O(t)})$ which explains the running time. We now prove the
  claim.

  Let $f:V \rightarrow L$ be an assignment of labels to the nodes whose
  cost is finite (such an assignment always exists since each terminal
  is in some independent set). Let $E' \subset E$ be the set of edges
  ``cut'' by this assignment; that is, $uv \in E'$ iff $f(u) \neq f(v)$.
  The cost of this assignment is equal to the weight of $E'$ since the
  metric is uniform and the labeling costs are $0$ or $\infty$. We argue
  that $E'$ is a feasible solution for the \UMulC instance. Suppose not.
  Then there are terminals $u,v$ such that $uv$ is an edge
  in the demand graph $H$ and $u,v$ belong to the same connected
  component of $G-E'$. The label $j = f(u)$ corresponds to a maximal
  independent set $I_j$ in $H$ which means that $v \not \in I_j$. Thus
  $f(v) \neq j$ since $c(v,j) = \infty$. Therefore, $u$ and $v$ are
  assigned different labels and cannot be in the same connected component.
  
  Conversely, let $E' \subset E$ be a feasible solution for \UMulC
  instance and let $V_1,\dots,V_{\ell}$ be vertex sets of the
  connected components of $G-E'$. Let $T_j$ be the terminals in
  $V_j$. Since, all pairs of terminals connected by an edge in $H$ are
  seperated in $G-E'$, $T_j$ must be an independent set in $H$. For
  each $T_j$, consider a maximal independent set in $H$ containing all
  the vertices of $T_j$; pick arbitrary one if more than one
  exists. Let this independent be $I_{i_j}$. We construct a labeling $f$
  by labeling all vertices of $V_j$ by label $i_j$.
  It is easy to see that all terminals are assigned a label corresponding
  to an  independent set in $H$ containing that terminal. 
  Hence, labeling cost is equal to
  zero. Also, all vertices corresponding to same connected component
  in $G-E'$ are assigned the same label. Hence, cost of the edges cut
  by the assignment $f$ is at most the cost of the edges in $E'$.
\end{proofof}

\section{UGC-based hardness of  approximation results for \DMulC}\label{sec:DMulC_hardness}
%\dirhardness*
\label{sec:hardness}
In this section we prove hardness of approximation for \DMulCH, in particular
Theorem~\ref{thm:hardness} relating the hardness of approximation to
the flow-cut gap. Recall that
$\alpha_H$ is the worst-case flow-cut gap (equivalently, the integrality
gap of the distance LP) for instances of \DMulCH. 

%\hardness*
We prove the theorem via a reduction to \MinCSP and the hardness 
result of Ene, Vondr\'ak and Wu \cite{EneVW13}. We note that the
result is technical and invovles several steps. This is partly due
to the fact that the theorem is establishing a meta-result. The theorem
of \cite{EneVW13} is in a similar vein. In particular \cite{EneVW13}
establishes that the hardness of \MinCSP depends on the integrality
gap of a specific LP formulation \BasicLP. Our proof is based on
establishing a correspondence between \DMulCH and a specific
constraint satsifaction problem \MinbHCSP where $\beta_H$ is 
constructed from $H$ (this is the heart of the reduction) and
proving the following properties:
\iffalse
We prove hardness of approximation for \DMulC using the hardness
result known for \MinCSP. Proof idea is as follows: given a bi-partite
demand graph $H$, we construct a set of predicates $\beta_H$
satisfying following properties:
\fi
\begin{compactitem}
\item[$(I)$] Establish approximation equivalence between \DMulCH and
  \MinbHCSP. That is, prove that each of them reduces to the other
  in an approximation preserving fashion.
%there is a cost preserving reduction
%  from/to \DMulC with demand graph $H$ (\DMulCH) to/from \MinCSP with
%  predicate set $\beta_H$(\MinbHCSP).
\item[$(II)$] Prove that if the flow-cut gap for \DMulCH (equivalently
  the integrality gap of \DistLP) is $\alpha_H$ then the integrality gap
  of \BasicLP for \MinbHCSP is also $\alpha_H$.
\iffalse
Distance LP for \DMulCH is equivalent to Basic LP for
  \MinbHCSP. That is, if $\Im$ is a \DMulCH instance and $\Ic$ is the
  equivalent \MinbHCSP instance, then given a solution $\bx$ to
  distance LP for $\Im$, there exists a solution $\bz$ with equal cost
  to Basic LP for $\Ic$ and vice versa.
\fi
\end{compactitem}
From $(I)$, we obtain that the hardness of approximation factor 
for \DMulCH and \MinbHCSP coincide. From $(II)$, we 
can apply the result in \cite{EneVW13} which shows that, assuming UGC,
the hardness of approximation for \MinbHCSP is the same as the
integrality gap of \BasicLP. Putting together these two claims give us
our desired result.

It is not straightforward to relate \DistLP for \DMulCH and \BasicLP
for \MinbHCSP directly.  \BasicLP appears to be stronger on first
glance.  In order to relate them we show that a seemingly strong LP
for \DMulC that we call \LabelLP is in fact no stronger than \DistLP.
It is surprising that this holds even when $H$ is not fixed graph
since the size of \LabelLP has an exponential dependence on the size
of $H$. In fact this can be seen as the key technical fact unerlying the
entire proof and is independently interesting since it is quite
different from the undirected graph setting.  It is much easier to
relate \LabelLP and \BasicLP. The rest of this section is organized
as follows. In Section~\ref{sec:labellp} we describe \LabelLP and prove its
equivalence with \DistLP. In Section~\ref{sec:csp} we describe \MinCSP and
\BasicLP and formally state the theorem of \cite{EneVW13} that we rely on.
We then subsequently describe our reduction from \DMulCH to \MinbHCSP and
complete the proof.

\subsection{\LabelLP and equivalence with \DistLP for \DMulC}
\label{sec:labellp}
In Section~\ref{sec:tK2_MulC}, we saw that if demand graph $H$ has
size $k$, then there is a labeling LP for \MulC (the undirected
problem) with size $\text{poly}(2^k, n)$ and integrality gap at most
$2$ which improves upon the integrality gap of \DistLP which can be
$\Omega(\log k)$. Here we describe a natural labeling LP for \DMulC
(\LabelLP), but in contrast to the undirected case, we show that it is
not stronger than \DistLP. We show this equivalence on an instance by
instance basis.

Let the demand graph be $H$ with vertex set $V_H=\{s_1,\dots,s_k\}$,
and the supply graph be $G=(V_G,E)$ with $n$ vertices. We will assume
here, for ease of notation, that $V_H \subset V_G$. Define a labeling
set $L = \{0,1\}^k$ which corresponds to all subsets of $V_H$.  We
interpret the labels in $L$ as $k$-length bit-vectors; if $\sigma \in
L$ we use $\sigma[i]$ to denote the $i$'th bit of $\sigma$. For two labels $\sigma,
\sigma' \in L$ we say $\sigma_1 \leq \sigma_2$ if $\forall i, 
\sigma_1[i] \leq \sigma_2[i]$. 
To motivate the
formulation consider any set of edges $E' \subseteq E$ that can be
cut. In $G' = G-E'$ we consider, for each $v \in V$, the reachability
information from each of the terminals $s_1,s_2,\ldots,s_k$. For each
$v$ this can be encoded by assigning a label $\sigma_v \in L$ where
$\sigma_v[i] = 1$ iff $v$ is reachable from $s_i$ in $G'$.  $E'$ is a
feasible solution if $s_i$ cannot reach $s_j$ whenever $(s_i,s_j)$ is
an edge of $H$. The goal of the formulation to assign labels to
vertices and to ensure that demand pairs are separated. An edge
$e = (u,v)$ is cut if there is some $s_i$ such that $s_i$ can reach
$u$ but $s_i$ cannot reach $v$. We add several
constraints to ensure that the label assignment is consistent.
The basic variables are $z_{v,\sigma}$ for each $v \in V_G$ and
$\sigma \in L$ which indicate whether $v$ is assigned the label $\sigma$.
We also a variable $x_e$ for each edge $e=(u,v) \in E_G$ that is derived
from the label assignment variables. 
We start with the basic constraints involving these variables and
then add additional variables that ensure consistency of the assignment.
\begin{compactitem}
  \item Each vertex is labelled by exactly one label. For $v \in V_G$, 
    $\sum_{\sigma \in L} z_{v,\sigma} = 1$.
  \item Vertex $s_i$ is reachable from $s_i$. For $s_i \in V_H$ and
    any $\sigma \in L$ such that $\sigma[i] = 0$, $z_{s_i,\sigma} = 0$
  \item Demand edges are separated. That is, if  $(s_i,s_j) \in E_H$, then 
    $s_j$ is not reachable from $s_i$. That is, 
    $z_{s_j,\sigma} = 0$ for any $\sigma$ where $\sigma[i] = 1$ and $(s_i,s_j) \in E_H$.
\end{compactitem}
For each edge $e=(u,v)$ we have variables of the form $z_{e,
  \sigma_1\sigma_2}$ where the intention is that $u$ is labeled
$\sigma_1$ and $v$ is labeled $\sigma_2$. To enforce consistency
between edge assignment variables and vertex assignment variables we
add the following set of constraints.
\begin{compactitem}
  \item  For $e = (u,v) \in E_G$, $z_{u,\sigma_1}  = \sum_{\sigma_2 \in L} z_{e,\sigma_1 \sigma_2}$ and $z_{v,\sigma_2}  =  \sum_{\sigma_1 \in L} z_{e,\sigma_1\sigma_2}$.
\end{compactitem}
Finally, the auxiliary variable $x_e$ indicates whether $e$ is cut.
\begin{compactitem}
  \item For $e = (u,v) \in E_G$, $x_e = 1$ if for some $i$, 
    $u$ is reachable from $s_i$ and $v$ is not reachable from $s_i$. 
  Then, $x_e = 1$ if $z_{e,\sigma_1,\sigma_2} = 1$ for 
  $\sigma_1 \not \leq \sigma_2$. We thus set $x_e = \sum_{\sigma_1,\sigma_2 \in L : \sigma_1 \not \leq \sigma_2} z_{e,\sigma_1\sigma_2}$.
\end{compactitem}

It is not hard to show that if one constraints all the variables
to be binary then the resulting integer program is valid formuation
for \DMulC. Note that the number of variables is exponential in 
$k = |V_H|$. Relaxing the integrality constraint of variables, we get Label-LP~\ref{fig:labellp}.

\begin{theorem}
  \label{thm:distlp-labellp}
  For any instance $G,H$ of \DMulCH, the optimum solution values for
  the formulations \LabelLP and \DistLP are the same both in the
  fractional and integral settings.
\end{theorem}

%\begin{figure}[htb]
%  \centering
%\begin{boxedminipage}{0.5\linewidth}
%\vspace{-0.2in}
%\begin{align*}
%& \textbf{\LabelLP}\\
%\min \quad & \sum_{e \in E} w_e x_e\span \span\\
%& \sum_{\sigma \in L} z_{v,\sigma}& =1 & \qquad v\in V, \sigma \in L\\
%& z_{s_i,\sigma} & = 0& \qquad s_i \in V_H, \sigma \in L, \sigma[i] = 0\\
%& z_{s_j,\sigma} & =  0 & \qquad \sigma \in L \mbox{~s.t~} \sigma[i] = 1 \mbox{~and~} (s_i,s_j) \in E_H\\
%& z_{u,\sigma_1} - \sum_{\sigma_2\in L} z_{e,\sigma_1 \sigma_2}&=0 & \qquad e=(u,v) \in E_G, \sigma_1 \in L\\
%& z_{v,\sigma_2}  - \sum_{\sigma_1 \in L} z_{e,\sigma_1 \sigma_2}&= 0 &\qquad e=(u,v) \in E_G, \sigma_2 \in L\\
%& x_e - \sum_{\sigma_1,\sigma_2 \in L : \sigma_1 \not \leq \sigma_2} z_{e,\sigma_1\sigma_2} & = 0 & \qquad e \in E_G\\
%&  z_{v,\sigma},z_{e,\sigma_1\sigma_2} & \leq 1& \qquad  v \in V_G,  e \in E_G, \sigma,\sigma_1,\sigma_2 \in L\\
%&  z_{v,\sigma},z_{e,\sigma_1\sigma_2} & \geq 0& \qquad  v \in V_G,  e \in E_G, \sigma,\sigma_1,\sigma_2 \in L
%\end{align*}
%\end{boxedminipage}
%  \caption{\LabelLP for \DMulC}
%  \label{fig:labellp}
%\end{figure}

\begin{wrapfigure}{r}{0.45\textwidth}
\vspace{-1em}
\begin{boxedminipage}{0.45\textwidth}
\vspace{-0.2in}
\begin{align*}
& \qquad\qquad\qquad\qquad\textbf{\LabelLP}\\
&\min \quad \sum_{e \in E} w_e x_e\\
& \quad\sum_{\sigma \in L} z_{v,\sigma} =1  \qquad v\in V_G, \sigma \in L\\
& \qquad z_{s_i,\sigma} \quad= 0 \qquad s_i \in V_H, \sigma \in L, \sigma[i] = 0\span\\
& \qquad z_{s_j,\sigma} \quad =  0   \qquad \sigma \in L,\sigma[i] = 1,(s_i,s_j) \in E_H\\
& \sum_{\sigma_2\in L} z_{e,\sigma_1 \sigma_2}=z_{u,\sigma_1}   \quad e=(u,v) \in E_G, \sigma_1 \in L\\
&  \sum_{\sigma_1 \in L} z_{e,\sigma_1 \sigma_2}= z_{v,\sigma_2}   \quad e=(u,v) \in E_G, \sigma_2 \in L\\
& \sum_{\sigma_1,\sigma_2 \in L :\sigma_1 \not \leq \sigma_2} z_{e,\sigma_1\sigma_2}  = x_e \quad e \in E_G\\
&  z_{v,\sigma},z_{e,\sigma_1\sigma_2} \quad \leq 1 \quad v \in V_G,  e \in E_G, \sigma,\sigma_1,\sigma_2 \in L \\
&  z_{v,\sigma},z_{e,\sigma_1\sigma_2}  \quad \geq 0\quad  v \in V_G,  e \in E_G, \sigma,\sigma_1,\sigma_2 \in L
\end{align*}
\end{boxedminipage}
\caption{\LabelLP for \DMulC}
 \label{fig:labellp}
% \vspace{-4em}
 \end{wrapfigure}

The formulation has similarities to the earth-mover LP for metric
labeling considered in \cite{KleinbergT02,ChekuriKNZ04} except that the
``distance'' between labels is not a metric. Define a cost function $c
: L \times L \rightarrow \{0,1\}$ as follows: $c(\sigma,\sigma') = 0$
if $\sigma \le \sigma'$ and $1$ otherwise. In fact given the basic
labeling variables $z_{v,\sigma}$ the other variables are decided in a
min-cost solution.  We explain this formally.

\noindent {\bf Interpreting Variables $z_{e,\sigma_1\sigma_2}$ and
  $x_e$ as flow:} Let $e = (u,v)$ be an edge in $G$. Consider a
directed complete bipartite digraph $B_{uv}$ with vertex set
$\Gamma_u=\{u_\sigma \mid \sigma \in L\}$ and $\Gamma_v=\{v_\sigma
\mid \sigma \in L\}$.  We assign cost $c(\sigma,\sigma')$ on the edge
$(u_\sigma,v_{\sigma'})$.  We assign a supply of $z_{u,\sigma}$ on the
vertex $u_\sigma$ and a demand of $z_{v,\sigma}$ on the vertex
$v_\sigma$.  The values $z_{e,\sigma_1\sigma_2}$ can be thought of as
flow from $u_{\sigma_1}$ to $v_{\sigma_2}$ satisfying the following
properties: (i) total flow out of $u_{\sigma_1}$ must be equal to the
supply $z_{u,\sigma_1}$ ($z_{u,\sigma_1} = \sum_{\sigma_2 \in L}
z_{e,\sigma_1 \sigma_2}$) (ii) total flow into $v_{\sigma_2}$ must be
equal to $z_{v,\sigma_2}$ ($z_{v,\sigma_2} = \sum_{\sigma_1 \in L}
z_{e,\sigma_1\sigma_2}$) (iii) flow is non-negative
($z_{e,\sigma_1\sigma_2}\geq 0$). The cost of the flow according to
$c$ is precisely $x_e$ ($=\sum_{\sigma_1 \not \leq \sigma_2}
z_{e,\sigma_1\sigma_2}$).
In particular, given an assignment of the
values of the labeling variables $z_{u,\sigma}$, $\sigma \in L$ and
$z_{v, \sigma'}$, $\sigma' \in L$ which can be thought of as two
distributions on the labels, the smallest value of $x_e$ that can be
achieved is basically the min-cost flow in $B_{uv}$ with supplies
and demands defined by the two distributions. In other words the
other variables are completely determined by the distributions
if one wants a minimum cost solution.

In the sequel we use $\bar{z}_u$ to denote the vector of assignment
value $z_{u,\sigma}$, $\sigma \in L$ and refer to $\bar{z}_u$ as the
distribution corresponding to $u$.  We present the high-level
reduction between the solutions of the two LP's and refer to
Section~\ref{sec:appendix-distlp-labellp-equivalence} for the full
proof.

\noindent {\bf From \LabelLP to \DistLP:} Let $(\bx,z)$ be a feasible
solution to \LabelLP for an instance $(G,H)$. This solution 
satisfies the following two conditions: (i) If $(s_i,s_j) \in E_H$, then
$\sum_{\sigma \in \{0,1\}^k: \sigma[i] = 1} z_{s_i,\sigma} = 1$ and
$\sum_{\sigma \in \{0,1\}^k: \sigma[i] = 1} z_{s_j,\sigma} = 0$. (ii)
For an edge $e=(u,v) \in E_G$, and any terminal $s_i$, 
$x_e \geq \sum_{\sigma \in \{0,1\}^k:
  \sigma[i] = 1} z_{v,\sigma} - \sum_{\sigma \in \{0,1\}^k: \sigma[i]
  = 1} z_{u,\sigma}$.

Suppose $(s_i,s_j) \in E_H$ and $s_i,a_1,\dots,a_t,s_j$ is a path from
$s_i$ to $s_j$ in $G$. Then, plugging in the above inequalities for
the edges of the path, we get $\sum_{e \in p} x_e \geq 1$. Hence,
$\bx$ is a feasible solution to \DistLP and has same cost as
$(\bx,\bz)$.

\noindent {\bf From \DistLP to \LabelLP:} Let $\bx$ be a feaisble
solution to \DistLP.  We obtain a label assignment $\bz'$ as follows.
For a vertex $u \in V_G$, let $d(s_1,u) \leq d(s_2,u) \leq\dots \leq
d(s_k,u)$ (if not, rename the terminals accordingly). Here, $d(u,v)$
denotes the shortest path distance from $u$ to $v$ as per lengths
$x_e$.  For $i \in [0,k]$, let $\sigma_i = 0^i1^{k-i}$. Then, set
$z_{u,\sigma_0}' = d(s_1,u)$, $z_{u,\sigma_k}' = 1-d(s_k,u)$ and for
$i \in [1,k-1], z_{u,\sigma_i}' = d(s_{i+1},u) - d(s_i,u)$. For
$\sigma \not \in \{\sigma_0,\dots,\sigma_k\}, z_{u,\sigma}' = 0$. 
Once the label assignment for vertices is defined, we obtain the values of other
variables by considering each edge $e=(u,v)$ and using the min-cost flow
between $\bar{z}'_u$ and $\bar{z}'_v$ as decribed earlier.
%For edge, $e=(u,v)$, variables $z_{e,\sigma_1\sigma_2}$ are defined such
%that the flow is of minimum cost. 
In Section~\ref{sec:appendix-distlp-labellp-equivalence}, we prove that
flow has cost ($=x_e'$) at most $\max_{i \in [1,k]} d(s_i,v) -
d(s_i,u)$ which is upper bounded by $x_e$. Hence, cost of solution
$(\bx',\bz')$ to \LabelLP is upper bounded by cost of $\bx$.

\iffalse
\begin{lemma}\label{lem:distance_label_lp}
For any \DMulC instance $(G,H)$, given a feasible solution $\bx$ to Distance LP~\ref{}, we can find a solution $(\bx',z)$ of Labelling LP~\ref{fig:dmulc_labelling_lp} of cost atmost as much as the cost of $\bx$ and vice versa. 
\end{lemma}
\fi

\subsection{Min-CSP and \BasicLP} \label{sec:csp}
Min-CSP refers to a minimization version of constration satisfaction
problems. We set up the formalism borrowed from \cite{EneVW13}.  Let
$L$ denote the set of labels. A real-valued function $f: L^i
\rightarrow \mathbb{R}$ has \emph{arity} $i$.  Let $\Gamma = \{\psi
\mid \psi: L^i \rightarrow [0,1] \cup \{\infty\}, i \leq k\}$ be the
set of functions defined on $L$ with arity atmost $k$ and range $[0,1]
\cup \{\infty\}$. Let $\beta \subset \Gamma$ be a finite subset of
$\psi$. These functions are also refered to as \emph{predicates}. $k$
denotes the arity and $L$ denotes the alphabet of $\beta$. Each
$\beta$ induces an optimization problem \MinCSP.

\begin{definition}
An instance of Min-$\beta$-CSP consists of the following:
\begin{compactitem}
\item A vertex set $V$ and a set of tuples $T \subset \cup_{i=1}^k V^i$.
\item A predicate $\psi_t\in \beta$ for each tuple $t \in T$ where cardinality of $t$ matches the arity of $\psi_t$.
\item A non-negative weight function over the set of tuples, $w: T \rightarrow \mathbb{R}^+$.
\end{compactitem}
The goal is to find a label assignment $\ell: V \rightarrow L$ to minimize $\sum_{t = (v_{i_1},\dots,v_{i_j})\in T} w_t \cdot \psi_t(\ell(v_{i_1}),\dots,\ell(v_{i_j}))$.
\iffalse
\begin{equation*}
\sum_{t = (v_{i_1},\dots,v_{i_j})\in T} w_t \cdot \psi_t(\ell(v_{i_1}),\dots,\ell(v_{i_j}))
\end{equation*}
\fi
\end{definition}

Consider an integer programming formulation with following variables: for
each vertex $v \in V$ and label $\sigma \in L$, we have a variables
$z_{v,\sigma}$ which is $1$ if $v$ is assigned label $\sigma$. Also, for
each tuple $t = (v_{i_1},\dots,v_{i_j}) \in T$ and $\alpha \in
L^{|t|}$, we have a boolean variable $z_{t,\alpha}$ which is $1$ if
$v_{i_p}$ is labelled $\alpha_p$ for $p \in [1,j]$. These variables
satisfy following constraints:
\begin{compactitem}
\item Each vertex $v$ is labelled exactly once: $\sum_{\sigma \in L} z_{v,\sigma} = 1$.
\item Variables $z_{v,\sigma}$ and $z_{t,\alpha}$ are consistent. That is, if $v \in t$ is assigned label $\sigma$, then $z_{t,\alpha}$ must be zero if $\alpha$ does not assign label $\sigma$ to $v$. For every touple $t \in T, v = t[i], \sigma \in L$, we have: $z_{v,\sigma} = \sum_{\alpha \in L^{|t|}: \alpha[i] = \sigma} z_{t,\alpha}$. 
\end{compactitem}
The objective is minimize $\sum_{t \in T} w_t \cdot \sum_{\alpha \in L^{|t|}} z_{t,\alpha}\cdot \psi_t(\alpha)$.

\iffalse
The objective is to minimize the following function:
\begin{equation*}
LP(\cI) = \min \sum_{t \in T} w_t \cdot \sum_{\alpha \in L^{|t|}} z_{t,\alpha}\cdot \psi_t(\alpha)
\end{equation*}
\fi

\BasicLP is the LP relaxation obtained by allowing the variables to
take on values in $[0,1]$ and is described in the figure.
For instance $\cI$, $OPT(\cI)$ and $LP(\cI)$ refer to the 
fractional and integral optimum values respectively.

%\begin{figure}[htb]
%  \centering
%\begin{boxedminipage}{0.5\linewidth}
%\vspace{-0.2in}
%\begin{align*}
%& \quad\textbf{\BasicLP}\span\span\span\\
%& \min \quad \sum_{t \in T} w_t \cdot \sum_{\alpha \in L^{|t|}} z_{t,\alpha}\cdot \psi_t(\alpha)\span \span \\
%& \sum_{\sigma \in L} z_{v,\sigma}& =1 & \qquad v\in V\\
%& z_{v,\sigma} -\sum_{\alpha \in L^{|t|}: \alpha[i] = \sigma} z_{t,\alpha}& = 0 & \qquad t \in T, v = t[i], \sigma \in L\\
%& z_{v,\sigma},z_{t,\alpha} & \geq 0 & \qquad v \in V, \sigma \in L, t\in T, \alpha \in L^{|t|}\\
%& z_{v,\sigma},z_{t,\alpha} & \leq 1 & \qquad v \in V, \sigma \in L, t\in T, \alpha \in L^{|t|}
%\end{align*}
%\end{boxedminipage}
%  \caption{Basic LP for Min-$\beta$-CSP}
%  \label{fig:basic_lp}
%\end{figure}

\begin{wrapfigure}{r}{0.45\textwidth}
\vspace{-2em}
\begin{boxedminipage}{0.45\textwidth}
\vspace{-0.2in}
\begin{align*}
& \quad\textbf{\BasicLP}\\
& \min \quad \sum_{t \in T} w_t \cdot \sum_{\alpha \in L^{|t|}} z_{t,\alpha}\cdot \psi_t(\alpha) \\
& \qquad \sum_{\sigma \in L} z_{v,\sigma} \quad =1  \qquad v\in V\\
& \sum_{\alpha \in L^{|t|}: \alpha[i] = \sigma} z_{t,\alpha}= z_{v,\sigma} \quad t \in T, v = t[i], \sigma \in L\\
& z_{v,\sigma},z_{t,\alpha}\quad \geq 0  \quad v \in V, \sigma \in L, t\in T, \alpha \in L^{|t|}\\
& z_{v,\sigma},z_{t,\alpha} \quad \leq 1  \quad v \in V, \sigma \in L, t\in T, \alpha \in L^{|t|}
\end{align*}
\end{boxedminipage}
  \caption{Basic LP for Min-$\beta$-CSP}
  \label{fig:basic_lp}
  \vspace{-3em}
\end{wrapfigure}

A particular type of predicate termed \NAE (for \emph{not all equal}) is
important in subsequent discussion. 
\begin{definition}
  For $i \geq 2$, $NAE_i: L^i \rightarrow \{0,1\}$ be a predicate such
  that $NAE_i(\sigma_1,\dots,\sigma_i) = 0$ if $\sigma_1 = \sigma_2 = \dots =
  \sigma_i$ and $1$ otherwise.
\end{definition}

The following theorem shows that the hardness of \MinbHCSP coincides
with the integrality gap of \BasicLP if $\NAE_2$ is in $\beta$.
\begin{theorem} (Ene, Vonrak, Wu~\cite{EneVW13}) 
  \label{thm:csp-hardness}
  Suppose we have a
  \MinbCSP instance $\cI = (V,T,\Psi_t, t \in T, w)$ with fractional
  optimum (of Basic LP) $LP(\cI) = c$, integral optimum $OPT(\cI) = s$,
  and $\beta$ contains the predicate $NAE_2$. Then, assuming UGC, for
  any $\epsilon$, for some $\lambda>0$, it is NP-hard to distinguish
  between instances of \MinbCSP where the optimum value is at least
  $(s-\epsilon)\lambda$ and instances where the optimum value is less
  than $(c+\epsilon)\lambda$ .
\end{theorem}

\subsection{\DMulCH and an equivalent \MinbCSP Problem}\label{sec:csp-reduction}
In this section, we show that given a bipartite directed graph $H = (S
\cup T, E_H)$, we can construct a set of predicates $\beta_H$ such
that \DMulCH is equivalent to \MinbHCSP. The notion of equivalence is
as follows. We give a reduction from instances of \DMulCH to instances
of \MinbHCSP which preserves the cost of optimal integral solution and
in addition also preserves the cost of optimum fractional solution
to \LabelLP and \BasicLP. Similarly we give a reduction from \MinbHCSP
to \DMulCH which preserves the cost of both the integral and
fractional solutions. 

\iffalse
These reductions, when combined with the
equivalence of \DistLP and \LabelLP, prove that the worst-case 
flow-cut gap for \DMulCH is the same as the worst-case integrality
gap of \BasicLP for \MinbHCSP and also that the approximation hardness
of \DMulCH and \ We can then invoke 
Theorem~\ref{thm:csp-hardness} to prove Theorem

that there is one to one correspondance between fractional/integral
solutions to the \LabelLP for instances of \DMulCH and
fractional/integral solutions to \BasicLP on instances of . This in
turn implies that integrality gap of Label LP and Basic LP are equal
and hardness of approximation of \DMulCH and \MinbHCSP are
equal. Theorem~\ref{} combined with these facts shows that hardness of
approximation of \DMulCH is at least the integrality gap of Label LP
which in turn is equal to integrality gap of distance
LP(Theorem~\ref{}).
\fi

\iffalse
The label set for $\beta_H$ is similar to the label set
for \LabelLP. In \LabelLP, there is a constraint that each terminal
$a_i$ gets the correct label ($z_{a_i,\sigma} = 0$ if $\sigma[i]
=0$). In $\beta_H$, we will have a predicate $\psi_i$ such that
$\psi_i(\sigma) = 0$ if $\sigma$ is the correct label for $a_i$ and
$\infty$ otherwise. Another part of the Label LP is that an edge $e =
(u,v)$ is charged if $\ell(u) \not \leq \ell(v)$($x_e \geq
\sum_{\sigma_1,\sigma_2 \in L:\sigma_1\not \leq
  \sigma_2}z_{e,\sigma_1\sigma_2}$). In $\beta_H$, we have a predicate
$\psi(\sigma_1,\sigma_2)$ which is $1$ if $\sigma_1 \not \leq
\sigma_2$ and $0$ othereise. And so on.
\fi

The basic idea behind the construction of $\beta_H$ from $H$ is to
simulate the constraints of \LabelLP via the predicates of
$\beta_H$. In addition to setting up $\beta_H$ correctly, we also need
to preprocess the supply graph to prove the correctness of the
reductions. Let the bipartite demand graph $H$ be $(S \cup T, E_H)$
with $S = \{a_1,\dots,a_p\}$ and $T=\{b_1,\dots,b_q\}$ as the
bipartition. For $u \in S$ let $N_H^+(u) = \{ v \in T \mid (u,v) \in
E_H\}$ be the neighbors of $u$ in $H$. For $i \in [1,p]$, let $Y_i =
\{j \in [1,p] \mid N_H^+(a_j) \subseteq N_H^+(a_i)\}$. That is, if
$a_j \in Y_i$, the set of terminals that $a_j$ needs to be separated
from is a subset of the terminals that $a_i$ needs to be separated
from.  For $j \in [1,q]$ let $Z_j = \{i \in [1,p] \mid a_i b_j \not
\in E_H\}$. That is, $Z_j$ is the set of all terminals in $S$ the do
not need to be separated from $b_j$.
\vspace{-4mm}
\paragraph{Assumptions on supply graph:} We will assume that the supply
graph $G$ in the instances of \DMulCH satisfy the following properties.
\begin{compactitem}
\item Assumption I: $G$ may contain \emph{undirected} edges.  The
  meaning of this is that a path may include this edge in either
  direction. A simple and well-known gadget shown in
  Fig~\ref{fig:undirected_edge_gadget} shows that this is without loss of
  generality.
\item Assumption II: For $1 \le j \le q$ and $i \in Z_j$,
  there is an infinite weight edge from $a_i$ to $b_j$ in $G$. Moreover $b_j$ has no outgoing edge.
\item Assumption III: For $1 \le i \le p$, and $i' \in Y_i$, there is an infinite weight edge from $a_{i'}$ to 
  $a_i$ in $G$. Moreover $a_i$ has no other incoming edges.  
\end{compactitem}

\begin{wrapfigure}{r}{0.3\textwidth}
\centering
\includegraphics{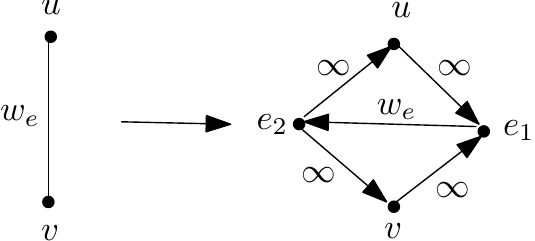}
\caption{Gadget to convert undirected edge/$NAE_2$ predicate to a directed graph}
\label{fig:undirected_edge_gadget}
\vspace{-1em}
\end{wrapfigure}

The preceding assumptions are to make the construction of $\beta_H$
and the subsequent proof of equivalence with \DMulCH somewhat more
transparent and technically easier.  Undirected edges allow us to use
the $\NAE_2$ predicate in $\beta_H$.  Assumption II and III simplify
the reachability information of terminals that needs to be kept track
of and this allows for a simpler label set definition and easier proof
of equivalence.

\DistLP easily generalizes to handle undirected edges; in examining
paths from $s_i$ to $t_i$ for a demand pair we allow an undirected
edge to be used in both directions. A more technical part is to
generalize \LabelLP to handle undirected edges in the supply graph.
For a directed edge $e$ recall that $x_e = \sum_{\sigma_1, \sigma_2
  \in L: \sigma_1 \not \leq \sigma_2} z_{e,\sigma_1\sigma_2}$. For an
undirected edge $e$ we set $x_e = \sum_{\sigma_1, \sigma_2 \in L: \sigma_1 \neq \sigma_2} z_{e,\sigma_1\sigma_2}$.
See Section~\ref{sec:appendix-theorem-label-basic-equivalence} for the
justification of the assumptions.

\vspace{-3mm}
\paragraph{Constructing $\beta_H$ from $H$:}
Next, we formally define $\beta_H$ for a bipartite graph $H = (S \cup
T, E_H)$ where $S=\{a_1,\ldots,a_p\}$ and $T =
\{t_1,\ldots,t_q\}$.. Recall the definitions of $Y_i$ for $1 \le i \le
p$ and $Z_j$ for $1 \le j \le q$ based on $E_H$.  Observe that no
vertex other than $b_j$ is reachable from $b_j$. And, since labels
encode the reachability from terminals, we can ignore the reachability
from $b_j$ and define $\beta_H$ with respect to terminal set $S$. For
$\sigma \in \{0,1\}^p$, let $J_{\sigma} = \{i \in [1,p] \mid \sigma[i]
= 1\}$

\begin{compactitem}
\item Alphabet (Label Set) $L = \{0,1\}^p$. Labels encode the list of
  $a_i$'s from which a vertex is reachable.
\item For $i \in [1,p]$, a unary predicate $\psiai$ encode the correct label for $a_i$ and is defined as follows: $\psiai(\sigma) = 0$ if $\Js= Y_i$,
otherwise $\psiai(\sigma) = \infty$.
\iffalse
\begin{equation*}
\psiai(\sigma) = \begin{cases}
0 &\Js= Y_i\\
\infty & otherwise
\end{cases}
\end{equation*}
\fi
\item For $j \in [1,q]$, predicate $\psibj$ that encodes the correct label
  for $b_j$. $\psibj(\sigma) = 0$ if $\Js= Z_j$, otherwise $\psibj(\sigma) = \infty$.
\iffalse
\begin{equation*}
\psibj(\sigma) = \begin{cases}
0 & \Js = Z_j\\
\infty & otherwise
\end{cases}
\end{equation*}
\fi
\item A binary predicate $\cC$ that encodes if a directed edge is cut or not.
If $\sigma_1 \le \sigma_2$ $\cC(\sigma_1,\sigma_2) = 0$, otherwise
$\cC(\sigma_1,\sigma_2) = 1$.
\iffalse
\begin{equation*}
\cC(\sigma_1,\sigma_2) = \begin{cases}
0 & \sigma_1\leq\sigma_2\\
1& otherwise
\end{cases}
\end{equation*}
\fi
\item A binary predicate $\NAE_2$ that encode if an undirected edge is cut or 
not. If $\sigma_1 = \sigma_2$  $\NAE_2(\sigma_1,\sigma_2) = 0$, otherwise
$\NAE_2(\sigma_1,\sigma_2) = 1$.
\iffalse
\begin{equation*}
NAE_2(\sigma_1,\sigma_2) = \begin{cases}
0 & \sigma_1=\sigma_2\\
1& otherwise
\end{cases}
\end{equation*}
\fi
\end{compactitem}

Thus $\beta_H = \{\cC,\NAE_2\} \cup \{\psiai \mid i \in [1,p]\}\cup
\{\psibj \mid j \in [1,q]\}$. \MinbHCSP has label set $L$, predicate
set $\beta_H$ and arity $2$.

The main technical theorem we prove is the following. We remark
that when we refer to \DMulCH we are referring to the problem
where the supply graph satisfies the assumptions I, II, III that
we outlined previously.

\begin{theorem}\label{thm:label_basic_equivalence}
  Let $H$ be a directed bipartite graph. 
  There is a polynomial time reduction that given 
  a \DMulCH instance $\Im=(G=(V_G,E_G,w_G:E_G \rightarrow
  R^+),H=(S\cup T,E_H))$, outputs a \MinbHCSP instance $\Ic =
  (V_C,T_C,\psitc:T_C\rightarrow \beta_H , w_{T_c}:T_C\rightarrow
  R^+)$ such that the following holds: given a solution $(\bx,\bz)$ of
  the Label LP for $\Im$, we can construct a solution $\bz'$ of Basic
  LP for $\Ic$ with cost at most that of $(\bx,\bz)$ and vice versa.
  More over, if $(\bx,\bz)$ is an integral solution, then $\bz'$ is
  also an integral solution and vice versa. A similar reduction
  exists from \MinbHCSP to \DMulCH.
\end{theorem}

With the preceding theorem in place we can formally prove
Theorem~\ref{thm:hardness}

\begin{proofof}{Theorem~\ref{thm:hardness}} Let $\Im$ be some fixed
  instance of \DMulCH with flow-cut gap $\alpha_H$. From
  Theorem~\ref{thm:distlp-labellp} the integrality gap of \LabelLP on
  $I_m$ is also $\alpha_H$. Let $I_C$ be the \MinbHCSP instance
  obtained via the reduction guaranteed by
  Theorem~\ref{thm:label_basic_equivalence}. $\Im$ and $\Ic$ have the
  same integral cost. Fractional cost of \LabelLP for $\Im$ and and
  \BasicLP for $\Ic$ are also the same. Therefore the integrality gap
  of \BasicLP on $\Ic$ is also $\alpha_H$. Via Theorem~\ref{thm:csp-hardness},
  assuming UGC, \MinbHCSP is hard to approximate within a factor of
  $\alpha_H - \eps$ for any fixed $\eps > 0$.

  Theorem~\ref{thm:label_basic_equivalence} (the second part) implies
  that \MinbHCSP reduces to \DMulCH in an approximation preserving
  fashion. Thus, \DMulCH is at least has hard to approximate as
  \MinbHCSP which implies that assuming UGC,
  the hardness of \DMulCH is at least $\alpha_H - \eps$.
\end{proofof}

\BasicLP and \LabelLP are almost identical except for the fact
that \LabelLP is defined with label set $\{0,1\}^k$ where $k=p+q$ is the
total number of terminals whereas \BasicLP is defined with
label set $\{0,1\}^p$. However, since $b_i$'s do not have any outgoing
edge, reachability from $b_i$ is trivial. The formal proof of
equivalence is long and somewhat tedious. We need to consider a
reduction from \MinbHCSP to \DMulCH and vice-versa. In each direction
we need to establish the equivalence of the cost of \LabelLP and \BasicLP
for both integral and fractional settings. We will briefly sketch the reduction here. Full proofs can be found in Section~\ref{sec:appendix-theorem-label-basic-equivalence}. 

\noindent {\bf Reduction from \MinbHCSP to \DMulCH:} Given a \MinbHCSP
instance $I_C$, equivalent \DMulCH instance $I_M$ is constructed as
follows: (i) Vertex set of $I_M$ is same as that of $I_C$ (ii) For $i
\in [1,p]$, name one of the vertex $v \in V_C$ with constraint
$\psiai(v)$ as vertex $a_i$ (iii) For constraint $\cC(u,v)$, add a
directed edge $e_t=(u,v)$ and for constraint $\NAE_2(u,v)$, add an undirected
edge $e_t=uv$ (iv) Add edges among $a_i$'s and
$b_j$'s so as to satisfy Assumption II and III.

Next, we show how to convert a soltion for one LP to a solution to other LP while preserving cost.

\noindent \textbf{From \LabelLP to \BasicLP:} Let $(\bx,\bz)$ be a
feasible solution to \LabelLP for $I_M$. Then, a feasible solution
$\bz'$ to \BasicLP for $I_C$ is simply a projection of $\bz$ from
label space $\{0,1\}^{p+q}$ to label space $\{0,1\}^p$.
Formally, $z'$
is defined as follows (i) For $\sigma \in \{0,1\}^p, z_{v,\sigma}' =
\sum_{\sigma' \in \{0,1\}^q} z_{v,\sigma\cdot\sigma'}$ (ii) For
$\sigma_1,\sigma_2 \in \{0,1\}^p, z'_{t,\sigma_1\sigma_2} =
\sum_{\sigma',\sigma''\in \{0,1\}^q} z_{e_t,\sigma_1 \cdot \sigma' \sigma_2 \cdot
  \sigma''}$. We can argue that $\sum_{\sigma_1, \sigma_2 \in
  \{0,1\}^p: \sigma_1 \not \leq \sigma_2} z'_{t,\sigma_1\sigma_2} \leq
\sum_{\sigma_3,\sigma_4 \in \{0,1\}^{p+q}: \sigma_3 \not \leq
  \sigma_4} z_{e_t,\sigma_3\sigma_4}=x_e$. Hence, cost of solution
$\bz'$ is at most the cost of solution $(\bx,\bz)$.

\noindent \textbf{From \BasicLP to \LabelLP:} Let $\bz$ be a feasible
solution to \BasicLP for $I_C$. Let $\sigma_0 = 1^q$, then a feaisble
solution $(\bx',\bz')$ to \LabelLP can be defined as an extension of
$\bz$ along $\sigma_0$. Formally, $\bz'$ is defined as follows: 
For $\sigma\in \{0,1\}^p,\sigma' \in \{0,1\}^q,v\in V_C,
z'_{v,\sigma\cdot \sigma'} = z_{v,\sigma}$ if $\sigma' = \sigma_0$
and $0$ otherwise. Similarly, for $\sigma_1,\sigma_2 \in \{0,1\}^p,
\sigma',\sigma'' \in \{0,1\}^q$,
$z'_{t,\sigma_1\cdot\sigma'\sigma_2\cdot\sigma''} =
z_{e_t,\sigma_1\sigma_2}$ if $\sigma' = \sigma'' = \sigma_0$ and $0$
otherwise. We prove that $x_e' = \sum_{\sigma_1,\sigma_2 \in
  \{0,1\}^{p+q}: \sigma_1 \not \leq \sigma_2} z'_{t,\sigma_1\sigma_2}
= \sum_{\sigma_3,\sigma_4 \in \{0,1\}^p: \sigma_3 \not \leq \sigma_4}
z_{e_t,\sigma_3\sigma_4}$. Hence, cost of solution $(\bx',\bz')$ is
equal to cost of solution $\bz$.

%\subsection{Hardness for Non-bipartite Demand graphs}
%\input{non-bi-partite-hardness-3.4}

\section{Approximating \DMulC }
\label{sec:DMulC_approximation}
We describe the algorithm that proves Theorem~\ref{thm:dir_approximation}.
Let $G=(V,E)$ and $H=(V,F)$ be the supply and demand graph for a given
instance of \DMulC. We provide a {\em generic}
randomized rounding algorithm that given a fractional solution $\bx$
to LP~\ref{fig:mulc_lp} for an instance $(G,H)$ of \DMulC returns a
feasible solution; the rounding does not depend on $H$. We can
prove that the returned solution is a $(k-1)$-approximation with respect
to the fractional solution $\bx$ or show that
$H$ contains an induced $k$-matching exension.
This algorithm is inspired by our recent rounding
scheme for \DMCut \cite{ChekuriM16}. The formal analysis can be found in
Section~\ref{sec:appendix_dirapproximation}.

\begin{algorithm}
	\begin{algorithmic}[1]
		\STATE Given a feasible solution $\bx$ to LP~\ref{fig:mulc_lp}
		\STATE For all $u,v \in V$, compute $d(u,v)$= shortest path length from $u$ to $v$ according to lengths $x_e$
		\STATE For all $u,v \in V$, compute $d_1(u,v) =\max(0,1-\min_{v' \in V, uv'\in F} d(v,v'))$ 
		\STATE Pick $\theta \in (0,1)$ uniformly at random
		\STATE $B_u = \{v \in V \mid d_1(u,v) \leq \theta\}$
                \STATE $E' = \cup_{u \in V} \delta^+(B_u)$
		\STATE Return $E'$
	\end{algorithmic}
	\caption{Rounding for \DMulC}
	\label{alg:directed_cut_rounding_scheme}
\end{algorithm}

The only subtelity in understanding the algorithm is the definition of
$d_1$ which we briefly explain.  Let $\bx$ be a feasible solution to
LP~\ref{fig:mulc_lp}. For $u,v \in V$, define $d(u,v)$ to be the
shortest path length in $G$ from vertex $u$ to vertex $v$ using
lengths $x_e$. We also define another parameter $d_1(u,v)$ for each
pair of vertices $u,v \in V$. $d_1(u,v)$ is the minimum non-negative
number such that if we add an edge $uv$ in $G$ with $x_{uv} =
d_1(u,v)$ then $u$ is still seperated from all the vertices it has to
be seperated from. Formally, for $u,v \in V$, $d_1(u,v) :=
\max(0,1-\min_{v' \in V, uv'\in F} d(v,v'))$.  If for some vertex $u$,
there is no demand edge leaving $u$ in $F$ then we define $d_1(u,v) =
0$ for all $v \in V$. Next, we do a simple ball cut rounding around
all the vertices as per $d_1(u,v)$. We pick a number $\theta \in
(0,1)$ uniformly at random. For all $u \in V$, we consider $\theta$
radius ball around $u$ for all $u \in V$; $B_u = \{v \in V \mid
d_1(u,v) \leq \theta\}$. And then cut all the edges leaving the set
$B_u$; $\delta^+(B_u) = \{vv' \in E_G \mid v \in B_u, v' \not \in
B_u\}$. Note that it is crucial that the same $\theta$ is used for all
$u$.

\bibliographystyle{plain}
\bibliography{multicut}

\appendix
\section{Proof of Theorem~\ref{thm:distlp-labellp}}\label{sec:appendix-distlp-labellp-equivalence}

\smallskip
\noindent
{\bf From \LabelLP to \DistLP:} Let $(\bx,z)$ be a feasible solution
to \LabelLP for the given instance of $G,H$. Consider a soltuion
$\bx'$ to \DistLP where we set $x_e' = x_e$. We claim that $\bx'$ is a
feasible solution to \DistLP for $G,H$. That is, for $(s_i,s_j) \in
E_H$, and a path $p$ from $s_i$ to $s_j$, we have $\sum_{e \in p} x_e'
\geq 1$.

\begin{lemma}\label{lem:x_e_lower_bound}
  For any edge $e=(u,v) \in E_G$ and $i \in \{1,\ldots,k\}$, 
  $x_e \geq \sum_{\sigma \in  L, \sigma[i] = 1} z_{u,\sigma} - \sum_{\sigma \in L, \sigma[i] =
    1} z_{v, \sigma}$.
\end{lemma}
\begin{proof}
  Recall the interpretation of variables $z_{e,\sigma_1\sigma_2}$ as
  flow from set $\Gamma_u = \{u_\sigma \mid \sigma \in L\}$ to
  $\Gamma_v = \{v_\sigma \mid \sigma \in L\}$. Consider the following
  partition of $\Gamma_u$ into $\Gamma_u^1 = \{u_\sigma \mid \sigma
  \in L, \sigma[i] = 1\}$ and $\Gamma_u^2 = \{u_\sigma \mid \sigma \in
  L, \sigma[i] = 0\}$. Similarly, consider the partition of $\Gamma_v$
  into $\Gamma_v^1$ and $\Gamma_v^2$. Amount of flow out of
  $\Gamma_u^1$ is equal to $\sum_{\sigma \in L, \sigma[i] = 1}
  z_{u,\sigma} $ and amount of flow coming into $\Gamma_v^1$ is equal
  to $\sum_{\sigma \in L, \sigma[i] = 1} z_{v,\sigma} $. Amount
  of flow from $\Gamma_u^1$ to $\Gamma_v^1$ is at most $\sum_{\sigma
    \in L, \sigma[i] = 1} z_{v,\sigma} $. Hence, flow from
  $\Gamma_u^1$ to $\Gamma_v^2$ is at least $\sum_{\sigma \in L,
    \sigma[i] = 1} z_{u,\sigma} -\sum_{\sigma \in L, \sigma[i] = 1}
  z_{v,\sigma}$. For $u_{\sigma_1} \in \Gamma_u^1, v_{\sigma_2} \in
  \Gamma_v^2$, we have $\sigma_1 \not \leq \sigma_2$ and hence,
\begin{equation*}
  x_e' = x_e =\sum_{\sigma_1,\sigma_2\in L: \sigma_1\not\leq \sigma_2} z_{e,\sigma_1\sigma_2} \geq \sum_{\sigma \in L: \sigma[i] = 1} z_{u,\sigma} -\sum_{\sigma \in L,:\sigma[i] = 1} z_{v,\sigma}.
\end{equation*}
\end{proof}

Let $(s_i,s_j) \in E_H$. We prove that for any path $p$ from $s_i$ to
$s_j$ in $G$ has $\sum_{e \in p} x_e' \geq 1$. Let the path $p$ be
$s_i,a_1,\dots,a_{\ell},s_j$. Then, by Lemma~\ref{lem:x_e_lower_bound}

\begin{eqnarray*}
x_{(s_i,a_1)} + \sum_{t= 1}^{\ell-1}x_{(a_{t},a_{t+1})} + x_{(a_\ell,s_j)} & \geq & \sum_{\sigma \in L: \sigma[i] = 1} \left((z_{s_i,\sigma} -  z_{a_1, \sigma}) + \sum_{t= 1}^{\ell-1} (z_{a_t,\sigma} -  z_{a_{t+1}, \sigma}) + (z_{a_\ell,\sigma} - z_{s_j,\sigma})\right)\\
& = & \sum_{\sigma \in L, \sigma[i] = 1} (z_{s_i,\sigma} - z_{s_j,\sigma})\\
\end{eqnarray*}

\LabelLP ensures that 
$z_{s_i,\sigma} = 0$ if
$\sigma[i] = 0$ and $z_{s_j,\sigma} = 0$ if $\sigma[i] = 1$. Hence,
$\sum_{\sigma \in L:\sigma[i] = 1} z_{s_i,\sigma} = 1$ and
$\sum_{\sigma \in L: \sigma[i] = 1} z_{s_j,\sigma} = 0$. Hence the
right hand side in the preceding inequality is $1$.

\smallskip
\noindent {\bf From \DistLP to \LabelLP:} Suppose $\bx$ is a feasible
solution to \DistLP for the given instance $G,H$.  We construct a
solution $(\bx',z)$ for \LabelLP such that $x_e' \leq x_e$ for all $e
\in E_G$. The edge lengths given by $\bx$ induce shortest path
distances in $G$ and we use $d(u,v)$ to denote this distance from $u$
to $v$. By adding dummy edges with zero cost as needed we can
assume that $d(u,v) \le 1$ for each vertex pair $(u,v)$.
With this assumption in place we have that for any edge $e=(u,v)$
and any terminal $s_i$, $d(s_i,v) \le  d(s_i,u) + x_e$; hence
$x_e \ge \max_{1 \le i \le k} (d(s_i,v) - d(s_i,u))$. We will in fact
prove that $x'_e \le \max_{1 \le i \le k} (d(s_i,v) - d(s_i,u))$.

We start by describing how to assign values to the
variables $z_{v,\sigma}$. Recall that these induce values to the
other variables if one is interested in a minimum cost solution.
%followed by $z_{e,\sigma_1\sigma_2}$ and finally describing
%$x'_e$. 
Let $d(u,v)$ denote the shortest distance from $u$ to $v$ in
$G$ as per lengths $x_e$.

For a vertex $u$, consider the permutation $\pi^u: \{1,\ldots,k\}
\rightarrow \{1,\ldots,k\}$ such that $d(s_{\pi^u(1)},u)\leq \dots
\leq d(s_{\pi^u(k)},u)$. In other words $\pi^u$ is an ordering of the
terminals based on distance to $u$ (breaking ties arbitrarily).
Define $\sigma_0^u,\dots,\sigma_{k}^u$ as follows:
\begin{equation*}
\sigma_i^u[j] = \begin{cases}1 & j \in \{\pi^u(1),\dots,\pi^u(i)\}\\
0 & j \not \in \{\pi^u(1),\dots,\pi^u(i)\}
\end{cases}
\end{equation*}
In the assignment above it is useful to interpret $\sigma^u_i$ 
as a set of indices of the terminals. Hence $\sigma^u_0$ corresponds to
$\emptyset$ and $\sigma^u_i$ to $\{\pi^{u}(1),\ldots,\pi^{u}(i)\}$.
Thus, these sets form a chain with.

The assignment of values to the variables $z_{u,\sigma}$, $\sigma \in L$
is done as follows:
\begin{equation*}
z_{u,\sigma} = \begin{cases}d(s_{\pi^u(1)},u) & \sigma = \sigma_0^u\\
d(s_{\pi^u(i+1)},u)-d(s_{\pi^u(i)},u) & \sigma = \sigma_i^u, 1 \leq i \leq k-1\\
1-d(s_{\pi^u(k)},u) & \sigma=\sigma_k^u\\
0 & \sigma \in L \setminus \{\sigma_0^u,\dots,\sigma_k^u\}
\end{cases}
\end{equation*}

\begin{lemma}\label{lem:assignment_variables_vertices}
$z_{u,\sigma}$ as defined above satisfy the following properties:
\begin{itemize}
\item $\forall u\in V_G, \sigma\in L, z_{u,\sigma} \geq 0$.
\item $\forall u\in V_G, \sum_{\sigma \in L} z_{u,\sigma} = 1$.
\item For $A \subseteq \{1,\dots,k\}$, define $\sigma_A\in L$ as: $\sigma_A[i] = 1$ for $i \in A$ and $0$ otherwise. Then, 
\begin{equation*}
\sum_{\sigma \geq \sigma_A} z_{u,\sigma} = 1 - \max_{i \in A} d(s_i,u)
\end{equation*}
\item Terminals are labelled correctly. That is, for 
  each $s_j$ and $\sigma\in L$, $z_{s_j,\sigma} = 0$ if $\sigma[j] = 0$.
\item If $(s_i,s_j) \in E_H$, then $z_{s_j,\sigma} = 0$ for $\sigma\in L$ such that $\sigma[i] = 1$. 
\end{itemize}
\end{lemma}
\begin{proof}
For $u \in V_G$, consider $\sigma_0^u,\sigma_1^u,\dots,\sigma^u_k$ as defined above. 
\begin{itemize}
\item $z_{u,\sigma} \geq 0$ is true by definition. 
\item By definition, $z_{u,\sigma} = 0$ if $\sigma \not \in \{\sigma_0^u,\dots,\sigma^u_k\}$. Hence, 
\begin{eqnarray*}
\sum_{\sigma \in L} z_{u,\sigma} & = & \sum_{i =0}^k z_{u,\sigma_i^u}\\
& = & d(s_{\pi^u(1)},u) + \sum_{i=1}^{k-1} d(s_{\pi^u(i+1)},u) - d(s_{\pi^u(i)},u) + 1-d(s_{\pi(k)},u)\\
& = & 1
\end{eqnarray*}
\item Let $j = \argmax_{i:\pi^u(i) \in A} d(s_{\pi^u(i)},u)$. Then, $\sigma^u_j,\dots,\sigma^u_k \geq \sigma_A$ and $\sigma^u_0,\dots,\sigma^u_{j-1} \not \geq \sigma_A$. Hence, 
\begin{eqnarray*}
\sum_{\sigma \geq \sigma_A} z_{u,\sigma} & = &\sum_{i=j}^k z_{u,\sigma^u_i}\\
& = & \sum_{i=j}^{k-1} d(s_{\pi^u(i+1)},u)-d(s_{\pi^u(i)},u) + 1-d(s_{\pi^u(k)},u)\\
& = & 1-d(s_{\pi^u(j)},u)\\
& = & 1-\max_{i:\pi^u(i)\in A} d(s_{\pi^u(i)},u)\\
& = & 1-\max_{i \in A} d(s_i,u)
\end{eqnarray*}
\item By definition of distance, $d(s_j,s_j) = 0$. Consider $A = \{j\}$. Applying the result from previous part, we get $\sum_{\sigma \geq \sigma_A} z_{s_j,\sigma} = 1-0=1$. Hence, $z_{s_j,\sigma} = 0$ if $\sigma \not \geq \sigma_A$. Equivalently speaking, $z_{s_j,\sigma} = 0$ if $\sigma[j] =0$. 
\item Let $(s_i,s_j) \in E_H$. Then, for the solution $\bx$ to be feasible, we must have $d(s_i,s_j) = 1$. Consider $A = \{i\}$. Then, using result from previous part, we get $\sum_{\sigma \geq \sigma_A} z_{s_j,\sigma} = 1 - 1 = 0$. Hence, $z_{s_j,\sigma} = 0$ if $\sigma \geq \sigma_A$. Equivalently speaking, $z_{s_j,\sigma} = 0$ if $\sigma[i] = 1$.
\end{itemize}
\end{proof}

Consider an edge $e=(u,v)$. Recall that once the distributions
of $\bar{z}_u$ and $\bar{z}_v$ are fixed then $x'_e$ is simply the 
min-cost flow between these two distributions in the 
digraph $B_{uv}$ with costs given by $c$. Our goal is to show that
this cost is at most $\max\{0,\max_{i} (d(s_i,v) - d(s_i,u))\}$.
Suppose we define a partial flow between $\bar{z}_u$ and $\bar{z}_v$
on zero-cost edges such that the total amount of this flow is $\gamma$
where $\gamma \in [0,1]$. Then it is easy to see that we can complete
this flow to achieve a cost of $(1-\gamma)$. This is because the
graph is a complete bipartite graph and costs are either $0$ or $1$
and $\bar{z}_u$ and $\bar{z}_v$ are distributions that have a total
of eactly one unit of mass on each side. 

Next, we define a partial flow of zero cost between $\bar{z}_u$ and
$\bar{z}_v$ by setting some variables $z_{e,\sigma_1\sigma_2}$ in a
greedy fashion as follows. Initially all flow values are zero.  For $i
= 0$ to $k$ in order we consider the vertex $u_{\sigma^u_i}$ with
supply $z_{u,\sigma^u_i}$. Our goal is to send as much flow as possible
from this vertex on zero-cost edges to demand vertices $v_{\sigma^v_j}$
which requires that $\sigma^u_i \le \sigma^v_j$. We maintain
the invariants that we do not exceed supply or demand in this process.
While trying to send flow out of  $u_{\sigma^u_i}$ we again use a greedy
process; if there are $j < j'$ such that
$\sigma^v_j$ and $\sigma^v_{j'}$ are both eligible to receive flow
on zero-cost edges and have capacity left, we use $j$ first; recall
that $\sigma^v_j$ corresponds to a subset of $\sigma^v_{j'}$. 

Let $z_{e,\sigma_1\sigma_2}$ be the partial flow created by the
algorithm. 

\begin{lemma}
  The total flow sent by the greedy algorithm described is at least 
  $1 -  \max\{0,\max_h (d(s_h,v) - d(s_h,u))\}$.
\end{lemma}

Assuming the lemma we are done because the zero-cost flow is
at least $1 - x_e$ and hence total cost of the flow is at most $x_e$.
This proves that $x'_e \le x_e$ as desired. We now prove the lemma.

Consider the greedy flow. Let $\ell$ be the maximum integer such that
$v_{\sigma^v_\ell}$ is not saturated by the flow.  If no such $\ell$
exists then the greedy algorithm has sent a total flow of one unit on
zero-cost edges and hence $x'_e = 0$.  Thus, we can assume $\ell$
exists. Moreover, in this case we can also assume that $\ell < k$ for
if $\ell = k$ the greedy algorithm can send more flow since
$\sigma^u_i \le \sigma^v_k$ for all $i$.  Let $\ell'$ be the maximum
integer such that $\sigma^u_{\ell'} \le \sigma^v_\ell$.  Such an
$\ell'$ exists since $\ell'=0$ is a candidate (corresponding to the
empty set). Moreover, $\ell' < k$ since $\sigma^u_k \not \leq \sigma^v_\ell$ 
since $\ell < k$. Let $\ell''$ be the minimum integer such that
$\sigma^u_{\ell'+1} \le \sigma^v_{\ell''}$. $\ell''$ exists because
$k$ is a candidate for it.

\begin{claim}
 $\pi_{\ell'+1}^u = \pi_{\ell''}^v$.
\end{claim}
\begin{proof}
By choice of $\ell, \ell', \ell''$ we have $\sigma^u_{\ell'} \le \sigma^v_{\ell}$
and $\sigma^u_{\ell'+1} \not \leq \sigma^v_{\ell}$ while
$\sigma^u_{\ell'+1} \leq \sigma^v_{\ell''}$. Thus $\ell'' \ge \ell+1$ and
$\sigma^u_{\ell'} \le \sigma^v_{\ell} \le \sigma^v_{\ell''-1}$.
Moreover, since $\ell''$ is chosen to smallest, $\sigma^u_{\ell'+1} \not \leq
\sigma^v_{\ell''-1}$. These facts imply the desired claim.
\end{proof}

We now claim several properties of the partial flow and justify them.
\begin{compactitem}
  \item $\forall i \in [0,\ell']$, $j \in [\ell+1,k]$, $z_{e,\sigma_i^u\sigma_j^v}  = 0$. This follows from the fact that the greedy algorithm did not 
    saturate $z_{v,\sigma^v_\ell}$.
  \item $\forall i \in [\ell'+1,k]$, $j \in [0,\ell''-1]$, $z_{e,\sigma_i^u\sigma_j^v} = 0$. From the definition of $\ell',\ell''$, this is not a zero cost edge.
    \item $\forall i \in [0,\ell']$, $\sum_{j = 0}^\ell
    z_{e,\sigma_i^u\sigma_j^v} = z_{u,\sigma_i^u}$. From definition of
    $\ell'$, for each $i \le \ell'$, there is a zero-cost edge from
    $u_{\sigma^u_i}$ to $v_{\sigma^v_\ell}$. Since the greedy
    algorithm did not saturate $v_{\sigma^v_\ell}$, it means that
    $u_{\sigma^u_i}$ is saturated and sends flow only to
    $v_{\sigma^v_1}, \ldots v_{\sigma^v_\ell}$.
  \item $\forall j \in [\ell'',k]$, $\sum_{i=\ell'+1}^k
    z_{e,\sigma_i^u\sigma_j^v} = z_{v,\sigma_j^v}$. By definition of
    $\ell$, for $j \geq \ell+1$ we have the property that
    $v_{\sigma_j^v}$ is saturated. As we argued above, for $i \in [\ell'+1,k], j \in [0,\ell''-1]$ we have $z_{e,\sigma_i^u\sigma_j^v} = 0$. Hence, for
    $j \geq \ell'' \geq \ell+1$, we have $\sum_{i=\ell'+1}^k
    z_{e,\sigma_i^u\sigma_j^v} = \sum_{i=0}^k
    z_{e,\sigma_i^u\sigma_j^v} = z_{v,\sigma_j^v}$.
\end{compactitem}

From the preceding claim we see that the total value of the partial flow
can be summed up as 
$$\sum_{\sigma_1,\sigma_2\in L} z_{e,\sigma_1\sigma_2} =  \sum_{i=0}^{\ell'} z_{u,\sigma_i^u}+\sum_{j=\ell''+1}^k z_{v,\sigma_j^v}.$$

Moreover, by construction of $\bar{z}_u$ and $\bar{z}_v$,
$$\sum_{i=0}^{\ell'} z_{u,\sigma_i^u} = d(s_{\pi^u(\ell'+1)},u)$$
and
$$\sum_{j=\ell''+1}^k z_{v,\sigma_j^v} = 1- d(s_{\pi^v(\ell'')},v).$$

Letting $h = \pi^u_{\ell'+1} = \pi^v_{\ell''}$ we see that from
the preceding equalities that the total flow routed on the
zero-cost edges is 
$$d(s_h,u) + 1 - d(s_h,v) = 1 - (d(s_h,v) - d(s_h,u)) \ge 1 - x_e.$$

This finishes the proof.

\section{Proof of Theorem~\ref{thm:label_basic_equivalence}}\label{sec:appendix-theorem-label-basic-equivalence}

The first two lemmas help establish that we can safely assume that
the supply graph satisfies the assumptions I, II, and III. We omit
the proof of the first lemma which involves tedious reworking of
some of the details on equivalence of \LabelLP and \DistLP.

\begin{lemma}
  For any instance $G,H$ of \DMulCH where the supply graph has
  undirected edges, the optimum solution values for the
  formulations \LabelLP and \DistLP are the same both in the fractional
  and integral settings.
\end{lemma}

Assuming the preceding lemma, following lemma is easy to prove:

\begin{lemma}
  For bipartite $H$, \DMulCH with a general supply graph and \DMulCH
  restricted to supply graphs satisfying Assumptions I, II and III are
  equivalent in terms of approximability and in terms of the
  integrality gap of \DistLP (equal to integrality gap of \LabelLP).
\end{lemma}
\begin{proof} 
  We sketch the proof. Undirected edges can be handled by the gadget
  shown in Fig~\ref{fig:undirected_edge_gadget}. It is easy to see
  that given any instance of \DMulCH with supply graph $G$ and
  bipartite demand graph $H$ we can first add dummy terminals to $G$
  and assume that each terminal $a_i$ has only one outgoing infinite
  weight edge (to the original terminal) and each $b_j$ has only one
  incoming infinite weight edge. With this in place adding edges to
  satisfy Assumptions II and III can be seen to not affect the 
  integral or fractional solutions to \DistLP.
\end{proof}

We will assume for simplicity that all weights (for edges and constraints)
are either $1$ or $\infty$. Generic weights can be easily simulated by
copies and the proofs make no essential use of weights other than
that some are finite and others are infinite.

%\colorcomment{At times it is refered to as constraint $\psiai(u)$ and at times it is refered to as $t = (u), \psitc(t) = \psiai$}

\subsubsection{Reduction from \MinbHCSP to \DMulCH}
Let the \MinbCSP instance be $I_C = (V_C, T_C, \psitc:T_C\rightarrow
\beta_H , \wtc:T_C\rightarrow R^+)$. We refer to touple $t = (u)$ with $\psitc(t) = \psiai$ as constraint $\psiai(u)$, $t = (u), \psitc(u) = \psibj$ as constraint $\psibj(u)$, $t = (u,v), \psitc(t) = \cC$ as constraint $\cC(u,v)$ and $t = (u,v), \psitc(u) = NAE_2$ as constraint $NAE_2(u,v)$. We assume that for every $i \in
[1,p]$, there is a constraint $\psiai(u_i)$ for some vertex $u_i \in
V_C$, and similarly for every $j \in [1,q]$ there is a constraint
$\psibj(v_j)$ for some vertex $v_j \in V_C$; moreover we will assume
that $u_1,\ldots,u_p,v_1,\ldots,v_q$ are distinct vertices. One can
ensure that this assumption holds by adding dummy vertices and dummy
constraints with zero weight. We create an instance $I_M =
(G=(V_G,E_G,w_G:E_G\rightarrow R^+), (S \cup T, E_H))$ of \DMulCH as
follows.  \iffalse For predicate $\cC$, add a directed edge, for
$NAE_2$, add an undirected edge and for predicate $\psiai$ and
$\psibj$, add terminals $a_i$ and $b_j$ respectively. Also, we add
more edges to satisfy Assumption II and III. Formally, reduction works
as follows: Equivalent \DMulCH instance $I_M =
(G=(V_G,E_G,w_G:E_G\rightarrow R^+), (S \cup T, E_H))$ is defined as
follows: \fi
\begin{itemize}
\item $V_G = V_C$, the vertex remains the same.  Pick vertices
  $u_1,\ldots,u_p$ and $v_1,\ldots,v_q$ that are all distinct such
  that for $1 \le i \le p$ there is a constraint $\psiai(u_i)$ in
  $\Ic$ and for $1 \le j \le q$ there is a constraint $\psibj(v_j)$ in
  $\Ic$.  This holds by our assumption. For $i \in [1,p]$ associate
  the terminal $a_i \in V_H$ with $u_i$ and for $j \in [1,q]$
  associate the terminal $b_j \in V_H$ with $v_j$.
\item $E_G$ and $w_G$ are defined as follows:
\begin{itemize}
\item For each constraint $\psiai(u)$ in $\Ic$ where $u \neq a_i$ add
  an undirected edge $e_t=a_iu$ to $E_G$ with $w_G(e_t) = \infty$.
\item For each constraint $\psibj(v)$ in $\Ic$ where $v \neq b_j$
  add an undirected edge $e_t = b_jv$ to $E_G$ with $w_G(e_t) = \infty$.
\item For each constraint $\cC(u,v)$ in $\Ic$ 
  add a directed edge $e_t=(u,v)$ in $G$ with $w_G(e_t)$ equal to the
  weight of the constraint in $\Ic$.
\item For each constraint $\NAE_2(u,v)$, add an undirected edge $e_t = uv$ 
with $w_G(e_t)$ equal to the weight of the constraint in $\Ic$.
\item For each $i \in [1,p]$ and for each $i' \in Y_i$, add a 
  directed edge $e = (a_{i'},a_i)$ with
  $w_G(e) = \infty$.
\item For each $j \in [1,q]$ and each $i \in Z_j$, 
  add a directed edge $e = (a_i,b_j)$ with $w_G(e) = \infty$. 
\end{itemize}
\end{itemize}

We now prove the equivalence of $\Ic$ and $\Im$ from the point of
view solutions to \BasicLP and \LabelLP respectively.

Given two labels $\sigma$ and $\sigma'$ which can be interpreted as
binary strings, we use the notation $\sigma \cdot \sigma'$ to denote
the label obtained by concatenating $\sigma$ and $\sigma'$.

\paragraph{\bf From \LabelLP to \BasicLP:} Suppose
$(\bx,\bz)$ is a feasible solution to \LabelLP for $\Im$.
We construct a solution $\bz'$ to \BasicLP for $\Ic$ 
in the following way. $\bz'$ is simply a projection of 
$\bz$ from label set $\{0,1\}^{p+q}$ onto label set $\{0,1\}^p$.
Recall that in the instance $\Im$ the terminals $b_1,\ldots,b_q$ 
do not have any outgoing edges. Hence, in the solution $(\bx,\bz)$ with
label space $\{0,1\}^{p+q}$, which encodes reachability from both
the $a_i$s and the $b_j$'s the information on reachability from
the $b_j$s does not play any essential role. We formalize this below. 
\begin{itemize}
\item For $v \in V_C, \sigma \in \{0,1\}^p$, $z'_{v,\sigma}= \sum_{\sigma'\in \{0,1\}^q} z_{v,\sigma\cdot\sigma'}$.
\item For unary constraint $t=(v) \in T_C$ and $\sigma \in \{0,1\}^p$, 
  $z'_{t,\sigma} = z'_{v,\sigma}$.
\item For binary constraint $t = (u,v) \in T_C$, for $\sigma_1,\sigma_2 \in \{0,1\}^p$,
\begin{equation*}
z_{t,\sigma_1\sigma_2}' = \sum_{\sigma' \in \{0,1\}^q} \sum_{\sigma'' \in \{0,1\}^q} z_{e_t,\sigma_1\cdot \sigma'\sigma_2\cdot\sigma''}
\end{equation*} 
\end{itemize}

Note that if $(\bx,\bz)$ is an integral solution then $\bz'$ as defined above is also an integral solution. 

Feasibility of $\bz'$ for \BasicLP is an ``easy'' consequence of the
projection operation but we prove it formally.
\begin{lemma}
  $\bz'$ as defined above is a feasible solution to \BasicLP for
  instance $I_C$.
\end{lemma}
\begin{proof}
 From the definition of $\bz'$, for each vertex $v$,
\begin{eqnarray*}
  \sum_{\sigma \in \{0,1\}^p} z_{v,\sigma}' & = & \sum_{\sigma\in \{0,1\}^p} \sum_{\sigma' \in \{0,1\}^q} z_{v,\sigma \cdot \sigma'}\\
  & = & 1
\end{eqnarray*}
which proves that one set of constraints holds.

Next, we prove that for $t \in T_C, v = t[i], \sigma \in L=\{0,1\}^p$,
the constraint $z_{v,\sigma}' -\sum_{\alpha \in L^{|t|}: \alpha[i] =
  \sigma} z_{t,\alpha}' =0$ holds. We consider unary and binary predicates
separately.
 
\begin{itemize} 
\item For $t = (v)$ s.t. $v = t[i], \sigma \in L=\{0,1\}^p$, 
\begin{equation*}
z_{v,\sigma}' - \sum_{\alpha \in L^{|t|}: \alpha[i] = \sigma} z_{t,\alpha}' =z_{v,\sigma}'- z_{t,\sigma}' = z_{v,\sigma}'-z_{v,\sigma}'=0.
\end{equation*}

\item For $t=(u,v) \in T_C, \sigma \in \{0,1\}^p$

\begin{eqnarray*}
z_{v,\sigma}' - \sum_{\sigma_1 \in \{0,1\}^p} z'_{t,\sigma_1\sigma} & = & \sum_{\sigma'\in \{0,1\}^p} z_{v,\sigma\cdot\sigma'} - \sum_{\sigma_1 \in \{0,1\}^p} \sum_{\sigma',\sigma''\in \{0,1\}^q} z_{e_t,\sigma_1\cdot \sigma''\sigma\cdot\sigma'}\\
& = & \sum_{\sigma'\in \{0,1\}^p} z_{v,\sigma\cdot\sigma'}  - \sum_{\sigma' \in \{0,1\}^q} \sum_{\sigma_1 \in \{0,1\}^p,\sigma''\in \{0,1\}^q} z_{e_t,\sigma_1\cdot\sigma''\sigma\cdot\sigma'}\\
& = & \sum_{\sigma'\in \{0,1\}^p} z_{v,\sigma\cdot\sigma'} - \sum_{\sigma'\in \{0,1\}^p} z_{v,\sigma\cdot\sigma'}\\
& = & 0.
\end{eqnarray*}
Similar argument holds for $u$ as well.
\end{itemize} 
\end{proof}

\begin{lemma}\label{lem:label_to_basic}
  The cost of $\bz'$ is at most $\sum_{e \in E_G} w_e x_e$ which is the
  cost of $(\bx,\bz)$ to $\Im$.
\end{lemma}

Before we prove Lemma~\ref{lem:label_to_basic} we establish
some properties satisfied by $(\bx,\bz)$.
\begin{lemma}\label{lem:label_sol_properties}
  If the solution $(\bx,\bz)$ to \LabelLP has finite cost, then the
  following conditions hold:
\begin{itemize}
\item For directed edge $e=(u,v)$, and for $i \in [1,p]$
  $x_e \geq \sum_{\sigma \in \{0,1\}^{p+q}: \sigma[i] = 1} z_{u,\sigma} - \sum_{\sigma \in
    \{0,1\}^{p+q}: \sigma[i] = 1} z_{v,\sigma}$. Hence, if edge $e$
  has infinite weight ($w_G(e) = \infty$), then $\sum_{\sigma \in
    \{0,1\}^{p+q}: \sigma[i] = 1} z_{u,\sigma} \leq \sum_{\sigma \in
    \{0,1\}^{p+q}: \sigma[i] = 1} z_{v,\sigma}$
\item For $i \in [1,p],\sigma \in \{0,1\}^p,\sigma' \in \{0,1\}^q$ s.t. $J_\sigma \neq Y_i$, we have $z_{a_i,\sigma \cdot \sigma'} = 0$. Hence, for $\sigma \in \{0,1\}^p, z_{a_i,\sigma}' = 1$ if $J_{\sigma} = Y_i$ and $0$ otherwise.
\item For $j \in [1,q],\sigma \in \{0,1\}^p,\sigma' \in \{0,1\}^q$ s.t. $J_\sigma \neq Z_j$ we have $z_{b_j,\sigma \cdot \sigma'} = 0$. Hence, for $\sigma \in \{0,1\}^p, z_{b_j,\sigma}' = 1$ if $J_{\sigma} = Z_j$ and $0$ otherwise.
\item For an undirected edge $e = uv \in E_G$ with $w_G(e) = \infty$, and $\sigma_1,\sigma_2 \in \{0,1\}^{p+q}$, $z_{e,\sigma_1\sigma_2} = 0$ if $\sigma_1 \neq \sigma_2$. For $\sigma \in \{0,1\}^{p+q}, z_{u,\sigma} = z_{v,\sigma}$ and for $\sigma_1\in \{0,1\}^p, z_{u,\sigma_1}' =z_{v,\sigma_1}'$. Hence, for $t = (u) \in T_C$ s.t. $\psitc(t) = \psiai, z'_{u,\sigma} = 1$ if $J_{\sigma} = Y_i$ and $0$ otherwise.
\end{itemize}
\end{lemma}

\begin{proof}
If $(\bx,\bz)$ has finite cost, then for an edge $e$ with infinite weight ($w_G(e) = \infty$), we must have $x_e = 0$.
\begin{itemize}
\item 
Let $e = (u,v)$ be a directed edge, and $i \in [1,p]$

\begin{eqnarray*}
x_e & = & \sum_{\sigma_1, \sigma_2 \in \{0,1\}^{p+q}: \sigma_1 \not \leq \sigma_2} z_{e,\sigma_1\sigma_2}\\
& \geq & \sum_{\sigma_1, \sigma_2 \in \{0,1\}^{p+q}: \sigma_1[i] =1, \sigma_2[i] = 0} z_{e,\sigma_1\sigma_2}\\
& = & \sum_{\sigma_1 \in \{0,1\}^{p+q}: \sigma_1[i] = 1} z_{u,\sigma_1}- \sum_{\sigma_1, \sigma_2 \in \{0,1\}^{p+q}: \sigma_1[i] =1, \sigma_2[i] = 1} z_{e,\sigma_1\sigma_2}\\
& \geq & \sum_{\sigma_1 \in \{0,1\}^{p+q}: \sigma_1[i] = 1} z_{u,\sigma_1} - \sum_{\sigma_2 \in \{0,1\}^{p+q}: \sigma_2[i] = 1} z_{v,\sigma_2}\\
& = & \sum_{\sigma \in \{0,1\}^{p+q}: \sigma[i] = 1} z_{u,\sigma} - \sum_{\sigma \in \{0,1\}^{p+q}: \sigma[i] = 1} z_{v,\sigma}
\end{eqnarray*}

 If edge $e$ has infinite weight, then $x_e = 0$ and $\sum_{\sigma \in \{0,1\}^{p+q}: \sigma[i] = 1} z_{u,\sigma} \leq \sum_{\sigma \in \{0,1\}^{p+q}: \sigma[i] = 1} z_{v,\sigma}$
\item We prove the following two statements which in turn imply that for $\sigma \in \{0,1\}^p, \sigma' \in \{0,1\}^q$, if $J_\sigma \neq Y_i$, then $z_{a_i,\sigma\cdot\sigma'} = 0$. 
\begin{eqnarray*}
\forall j \in Y_i, \sum_{\sigma\in\{0,1\}^p\sigma'\in\{0,1\}^q: \sigma[j] = 1} z_{a_i,\sigma\cdot \sigma'} &=& 1\\
\forall j \in [1,p] \setminus Y_i, \sum_{\sigma\in\{0,1\}^p\sigma'\in\{0,1\}^q: \sigma[j] = 1} z_{a_i,\sigma\cdot\sigma'} &=& 0
\end{eqnarray*}
Let $j \in Y_i$, then by construction of $G$, there exists an infinite weight edge from $a_j$ to $a_i$. Using the result from previous part we get 
\begin{equation*}
\sum_{\sigma \in \{0,1\}^p\sigma' \in \{0,1\}^q: \sigma[j] = 1} z_{a_i,\sigma\cdot \sigma'} \geq \sum_{\sigma \in \{0,1\}^p\sigma' \in \{0,1\}^q: \sigma[j] = 1} z_{a_j,\sigma\cdot \sigma'} 
\end{equation*}
\LabelLP enforces that term on the right side is lower bounded by
$1$ ($a_j$ reachable from itself). Hence, term on the left side is
lower bounded by $1$. Since, it is also upper bounded by $1$, it must
be equal to $1$.

Let $j \in [1,p]\setminus Y_i$. By definition of $Y_i$, we have
$N_H^+(a_j) \not \subseteq N_H^+(a_i)$. That is, there exists $j' \in
[1,q]$ such that $a_jb_{j'} \in E_H$ and $a_ib_{j} \not \in
E_H$. Since $a_jb_{j'} \in E_H$, \LabelLP enforces that
\begin{equation*}
\sum_{\sigma \in \{0,1\}^p\sigma' \in \{0,1\}^q: \sigma[j] = 1} z_{b_{j'},\sigma\cdot \sigma'}=0
\end{equation*}

Also, we have $a_i b_{j'} \not \in E_H$ and hence, there is
an infinite weight edge from $a_i$ to $b_{j'}$ in $G$. 
Applying the result from previous part, we get

\begin{equation*}
\sum_{\sigma \in \{0,1\}^p\sigma' \in \{0,1\}^q: \sigma[j] = 1} z_{a_i,\sigma\cdot \sigma'}\leq \sum_{\sigma \in \{0,1\}^p\sigma' \in \{0,1\}^q: \sigma[j] = 1} z_{b_{j'},\sigma\cdot \sigma'}=0
\end{equation*}

Next, to prove that $z_{a_i,\sigma}' = 1$ if $J_\sigma = Y_i$ and $0$ otherwise, we argue as follows:
\begin{eqnarray*}
1 &=& \sum_{\sigma \in \{0,1\}^p,\sigma' \in \{0,1\}^q} z_{a_i,\sigma\cdot\sigma'}\\
& = &\sum_{\sigma \in \{0,1\}^p: J_\sigma = Y_i} \sum_{\sigma' \in \{0,1\}^q} z_{a_i,\sigma\cdot\sigma'}+\sum_{\sigma \in \{0,1\}^p: J_\sigma \neq Y_i} \sum_{\sigma' \in \{0,1\}^q} z_{a_i,\sigma\cdot \sigma'}\\
& = & \sum_{\sigma \in \{0,1\}^p: J_\sigma = Y_i} z_{a_i,\sigma}'
\end{eqnarray*}
\item Again, we prove the following two statements which in turn implies that for $\sigma \in \{0,1\}^p,\sigma' \in \{0,1\}^q$ if $J_\sigma \neq Z_j$, then $z_{b_j,\sigma\cdot\sigma'} = 0$:
\begin{eqnarray*}
\forall i \in Z_j, \sum_{\sigma\in\{0,1\}^p\sigma'\in\{0,1\}^q: \sigma[i] = 1} z_{b_j,\sigma\cdot \sigma'} &=& 1\\
\forall i \in [1,p] \setminus Z_j, \sum_{\sigma\in\{0,1\}^p\sigma'\in\{0,1\}^q: \sigma[i] = 1} z_{b_j,\sigma\cdot\sigma'} &=& 0
\end{eqnarray*}

Let $i \in Z_j$. Hence, there is an infinite weight directed edge from
$a_i$ to $b_j$ in $G$. Applying the result from first part, we get

\begin{equation*}
\sum_{\sigma\in\{0,1\}^p\sigma'\in\{0,1\}^q: \sigma[i] = 1} z_{b_j,\sigma\cdot \sigma'} \geq \sum_{\sigma\in\{0,1\}^p\sigma'\in\{0,1\}^q: \sigma[i] = 1} z_{a_i,\sigma\cdot \sigma'}
\end{equation*}

\LabelLP enforces that right side is lower bounded by $1$ ($a_i$
reachable from itself). Hence, left side is lower bounded by $1$. It
is also upper bounded by $1$ and hence, it must be equal to $1$.

Let $i \in [1,p] \setminus Z_j$. Then, $a_ib_j \in E_H$ and hence, from the constraint in \LabelLP
\begin{equation*}
\sum_{\sigma\in\{0,1\}^p\sigma'\in\{0,1\}^q: \sigma[i] = 1} z_{b_j,\sigma\cdot \sigma'} = 0
\end{equation*}

Next, to prove that $z_{b_j,\sigma}' = 1$ if $J_\sigma = Z_j$ and $0$ otherwise, we argue as follows:
\begin{eqnarray*}
1 &=& \sum_{\sigma \in \{0,1\}^p,\sigma' \in \{0,1\}^q} z_{b_j,\sigma\cdot\sigma'}\\
& = &\sum_{\sigma \in \{0,1\}^p: J_\sigma = Z_j} \sum_{\sigma' \in \{0,1\}^q} z_{b_j,\sigma\cdot\sigma'}+\sum_{\sigma \in \{0,1\}^p: J_\sigma \neq Z_j} \sum_{\sigma' \in \{0,1\}^q} z_{b_j,\sigma\cdot\sigma'}\\
& = & \sum_{\sigma \in \{0,1\}^p: J_\sigma = Z_j} z_{b_j,\sigma}'
\end{eqnarray*}

\item For an undirected edge $e = uv$, $x_e = \sum_{\sigma_1,\sigma_2
    \in \{0,1\}^{p+q}:\sigma_1 \neq \sigma_2}
  z_{e,\sigma_1\sigma_2}$. Since, weight of $e$ is infinite, $x_e$
  must be $0$. Hence, $z_{e,\sigma_1\sigma_2} = 0$ if $\sigma_1 \neq
  \sigma_2$. For $\sigma_1 \in \{0,1\}^{p+q}$

\begin{eqnarray*}
z_{u,\sigma_1} &=& \sum_{\sigma_2 \in \{0,1\}^{p+q}} z_{e,\sigma_1\sigma_2}\\
 &=& z_{e,\sigma_1\sigma_1}\\
 & = &  \sum_{\sigma_2 \in \{0,1\}^{p+q}} z_{e,\sigma_2\sigma_1}\\
 & = & z_{v,\sigma_1}
\end{eqnarray*} 

Let $t = (u) \in T_C$ s.t. $\psitc(t) = \psiai$. If $u = a_i$, then we have already proved that $z'_{u,\sigma} = 1$ if $J_\sigma = Y_i$ and $0$ otherwise. If $u \neq a_i$, then there is an infinite weight undirected edge between $u$ and $a_i$ in $G$. Hence, $z_{u,\sigma}' = z_{a_i,\sigma}'$ for all $\sigma \in \{0,1\}^p$ and the result follows.
\end{itemize}
\end{proof}

\begin{proofof}{Lemma~\ref{lem:label_to_basic}} 
Next, we argue about the cost of the solution $\bz'$. We assume here
that $(\bx,\bz)$ has finite cost.
For a constraint $t \in T_C$, the cost according to $\bz'$ is
$\wtc(t) \sum_{\alpha \in L^{|t|}} z'_{t,\alpha}\cdot \psi_t(\alpha)$.
We consider four cases based on the type of $t$.

\begin{itemize}
\item $t$ corresponds to constraint of the form $\psiai(v)$. As argued
  in Lemma~\ref{lem:label_sol_properties}, then $z'_{t,\sigma} =
  z_{v,\sigma} = 0$ if $J_\sigma \neq Y_i$ and $1$ if $J_{\sigma} =
  Y_i$. On the other hand, $\psiai(\sigma) = 0$ if $J_\sigma = Y_i$ and $\infty$ if
  $J_{\sigma} \neq Y_i$. Hence, $z'_{t,\sigma}\psiai(\sigma) = 0$
  for all $\sigma$. Therefore this constraint contributes zero to the cost.
\item $t$ corresponds to constraint of the form $\psibj(v)$.  From
  Lemma~\ref{lem:label_sol_properties}, $z_{t,\sigma}' = z_{v,\sigma}'
  = 0$ if $J_\sigma \neq Z_j$ and $1$ if $J_\sigma =
  Z_j$. And $\psibj(\sigma) = 0$ if $J_\sigma = Z_j$ and $\infty$ if
  $J_{\sigma} \neq Z_j$. Hence, $z_{t,\sigma}'\psibj(\sigma) = 0$ for
  all $\sigma$. Therefore the contribution of this constraint is zero.
\item $t$ corresponds to constraint $\cC(u,v)$. This corresponds to
a directed edge $e_t=(u,v)$ in $G$ and the cost paid by $(\bx,\bz)$ is $x_{e_t}$.
The cost for $t$ in $\bz'$ is given by:
\begin{eqnarray*}
\sum_{\sigma_1,\sigma_2 \in \{0,1\}^p} z'_{t,\sigma_1\sigma_2} \cdot \cC(\sigma_1,\sigma_2) & = &\sum_{\sigma_1,\sigma_2 \in \{0,1\}^p: \sigma_1 \not \leq \sigma_2} z'_{t,\sigma_1\sigma_2}\\
& = & \sum_{\sigma_1,\sigma_2 \in \{0,1\}^p: \sigma_1 \not \leq \sigma_2} \sum_{\sigma',\sigma'' \in \{0,1\}^q} z_{e_t,\sigma_1\cdot\sigma'\sigma_2\cdot\sigma''}\\
& \leq & \sum_{\sigma_1,\sigma_2 \in \{0,1\}^p, \sigma',\sigma'' \in \{0,1\}^q :\sigma_1\cdot\sigma' \not \leq \sigma_2 \cdot \sigma''} z_{e_t,\sigma_1\cdot\sigma'\sigma_2\cdot\sigma''}\\
& = & x_{e_t}.
\end{eqnarray*}
First equality follows from the fact that $C(\sigma_1,\sigma_2) = 0$
if $\sigma_1 \leq \sigma_2$ and $1$ otherwise. Penultimate inequality
follows because if $\sigma_1 \not \leq \sigma_2$, then
$\sigma_1\cdot\sigma' \not \leq \sigma_2 \cdot \sigma''$ for any
$\sigma',\sigma'' \in \{0,1\}^q$.

\item $t$ corresponds to constraint $\NAE_2(u,v)$. This corresponds to
an undirected edge $e_t=uv$ in $G$ and the cost paid by $(\bx,\bz)$ is 
$x_{e_t}$. The cost for $t$ in $\bz'$ is given by:
\begin{eqnarray*}
\sum_{\sigma_1,\sigma_2 \in \{0,1\}^p} z'_{t,\sigma_1\sigma_2} \cdot \NAE_2(\sigma_1,\sigma_2) &=& \sum_{\sigma_1,\sigma_2 \in \{0,1\}^p: \sigma_1 \neq \sigma_2} z'_{t,\sigma_1\sigma_2} \\
& = & \sum_{\sigma_1,\sigma_2 \in \{0,1\}^p: \sigma_1 \neq \sigma_2}\sum_{\sigma',\sigma'' \in \{0,1\}^q} z_{e_t,\sigma_1\cdot \sigma'\sigma_2\cdot\sigma''}\\
& \leq & \sum_{\sigma_1,\sigma_2 \in \{0,1\}^p,\sigma',\sigma'' \in \{0,1\}^q: \sigma_1\cdot \sigma' \neq \sigma_2\cdot\sigma''}  z_{e_t,\sigma_1\cdot \sigma'\sigma_2\cdot\sigma''}\\
& = & x_{e_t}.
\end{eqnarray*}
\end{itemize}

Combining the four cases, the total cost of the
solution $\bz'$ is equal to the cost of the binary constraints
each of which corresponds to an edge in $G$ with the same weight.
From the above inequalities we see that the cost is atmost
$\sum_{e \in E_G} w_G(e)x_e$ which is the cost
of $(\bx,\bz)$. 
\end{proofof}

\paragraph{From \BasicLP to \LabelLP:} Let $\bz$ be a \BasicLP solution
to $\Ic$. Let $\sigma_0 = 1^q$. We define a solution $(\bx',\bz')$ 
to \LabelLP for $\Im$ as follows:
\begin{itemize}
\item For $v \in V_C, \forall \sigma_1 \in \{0,1\}^p, \sigma_2 \in \{0,1\}^q$,
\begin{equation*}
z'_{v,\sigma_1\cdot\sigma_2} = \begin{cases}z_{v,\sigma_1} & \sigma_2 = \sigma_0\\
0 & otherwise
\end{cases}
\end{equation*}
\item For unary constraint $t = (u)$ s.t. $u \not \in \{a_1,\dots,a_p,b_1,\dots,b_q\}$ and $\sigma_1,\sigma_2 \in \{0,1\}^p, \sigma_3,\sigma_4 \in \{0,1\}^q$,
\begin{equation*}
z_{e_t,\sigma_1\cdot\sigma_3\sigma_2\cdot\sigma_4}'=\begin{cases}
z_{u,\sigma_1} & \sigma_1 = \sigma_2, \sigma_3 = \sigma_4 = \sigma_0\\
0 & otherwise
\end{cases}
\end{equation*}
\item For binary constraint $t = (u,v) \in T_C$ such that $\psitc(t) = \cC$ or $\NAE_2$, and $\sigma_1,\sigma_2 \in \{0,1\}^p, \sigma_3,\sigma_4 \in \{0,1\}^q$
\begin{equation*}
z'_{e_t,\sigma_1\cdot\sigma_3\sigma_2\cdot\sigma_4} = \begin{cases}
z_{t,\sigma_1\sigma_2} & \sigma_3 = \sigma_4 = \sigma_0\\
0 & otherwise
\end{cases}
\end{equation*}
\item The edge variables $x_e'$ are induced by the $z'$ variables. We 
explicitly write them down.
For directed edge $e \in E_G$, $x_e' = \sum_{\sigma_1,\sigma_2 \in \{0,1\}^{p}, \sigma_3,\sigma_4 \in \{0,1\}^q: \sigma_1\cdot\sigma_3 \not \leq \sigma_2\cdot\sigma_4} z'_{e,\sigma_1\cdot\sigma_3\sigma_2\cdot\sigma_4}$. For undirected edge $e \in E_G$, $x_e' = \sum_{\sigma_1,\sigma_2 \in \{0,1\}^{p},\sigma_3,\sigma_4 \in \{0,1\}^q: \sigma_1\cdot \sigma_3 \neq \sigma_2\sigma_4} z'_{e,\sigma_1\cdot\sigma_3\sigma_2\cdot\sigma_4}$
\end{itemize}

It is easy to check that $(\bx',\bz')$ is integral if $\bz$ is integral.

\begin{lemma}
   $(\bx',\bz')$ is a feasible solution to \LabelLP for $\Im$. 
\end{lemma}
\begin{proof}
 It is easy to check that all the variables are non-negative
 and upper bounded by $1$.

 We show that the other constraints are satisfied one at a
 time. Recall that \LabelLP considered here has a constraint for
 undirected edges in addition to the constraints showed in
 Fig~\ref{fig:labellp}. The label set for \LabelLP is
 $\{0,1\}^{p+q}$ which we can write as $\{\sigma_1 \cdot \sigma_2 | \sigma_1 \in \{0,1\}^p, \sigma_2 \in \{0,1\}^q\}$.
\begin{itemize}
\item[Constraint 1:] For each $v$, $\sum_{\sigma \in \{0,1\}^{p+q}} z_{v,\sigma}' = 1$
\begin{eqnarray*}
\sum_{\sigma_1 \in \{0,1\}^p\sigma_2\in \{0,1\}^q} z_{v,\sigma_1\cdot\sigma_2}' & = & \sum_{\sigma_1 \in \{0,1\}^p} z_{v,\sigma_1\cdot\sigma_0}'  \\
& = & \sum_{\sigma_1 \in \{0,1\}^p} z_{v,\sigma_1}\\
& = & 1
\end{eqnarray*}
\item[Constraint 2:] For $\sigma_1 \in \{0,1\}^p, \sigma_2 \in \{0,1\}^q$, $z_{a_i,\sigma_1\sigma_2}' = 0$ if $\sigma_1[i] = 0$. And $z_{b_j,\sigma_1\sigma_2}' = 0$ if $\sigma_2[j] = 0$. 

  There is $t = (a_i) \in T_C$ such that $\psitc(t) = \psiai$. For
  $\bz$ to be a finite valued solution, we must have $z_{t,\sigma_1} =
  z_{a_i,\sigma_1} =0$ if $J_{\sigma_1} \neq Y_i$. Since, $i \in Y_i$,
  we have that $z_{a_i,\sigma_1} = 0$ if $\sigma_1[i] = 0$. And hence, 
  $z_{a_i,\sigma_1\sigma_2}' = 0$ if $\sigma_1[i] = 0$.

For $v \in V_C, z'_{v,\sigma_1\sigma_2} = 0$ if $\sigma_2 \neq \sigma_0 = 1^q$. Hence, $z'_{v,\sigma_1\sigma_2} = 0$ if $\sigma_2[j] = 0$. In particular, $z'_{b_j,\sigma_1\sigma_2} = 0$ if $\sigma_2[j] = 0$. 
\item[Constraint 3:] For $e = (u,v) \in E_G, \sigma_1 \in \{0,1\}^p,\sigma_3 \in \{0,1\}^q$, $z_{u,\sigma_1\cdot\sigma_3}' = \sum_{\sigma_2 \in \{0,1\}^p\sigma_4 \in \{0,1\}^q} z_{e,\sigma_1\cdot\sigma_3\sigma_2\cdot\sigma_4}$. If $\sigma_3 \neq \sigma_0$, then all the terms are zero and hence, the equality holds. Else, $\sigma_3 = \sigma_0$ and there are two types of edges:
\begin{itemize}
\item For $t = (u)$, $e = e_t$
\begin{eqnarray*}
\sum_{\sigma_2 \in \{0,1\}^p, \sigma_4 \in \{0,1\}^q} z_{e,\sigma_1\cdot\sigma_0\sigma_2\cdot\sigma_4}' & = & \sum_{\sigma_2 \in \{0,1\}^p} z_{e,\sigma_1\cdot\sigma_0\sigma_2\cdot\sigma_0}'\\
& = & z_{u,\sigma_1}\\
& = & z_{u,\sigma_1\cdot\sigma_0}' = z_{u,\sigma_1\cdot\sigma_3}'\\
\end{eqnarray*}
\item For $t = (u,v), e = e_t$, 
\begin{eqnarray*}
\sum_{\sigma_2 \in \{0,1\}^p, \sigma_4 \in \{0,1\}^q} z_{e,\sigma_1\cdot\sigma_0\sigma_2\cdot\sigma_4}' & = & \sum_{\sigma_2 \in \{0,1\}^q } z_{e,\sigma_1\cdot\sigma_0\sigma_2\cdot\sigma_4}'\\
 & = &\sum_{\sigma_2\in \{0,1\}^q} z_{e,\sigma_1\sigma_2} \\
& = & z_{u,\sigma_1}\\
& = & z_{u,\sigma_1\cdot\sigma_0}' = z_{u,\sigma_1\cdot\sigma_3}'
\end{eqnarray*}
\item[Constraint 4:] For $e = (u,v) \in E_G, \sigma_2 \in \{0,1\}^p,\sigma_4 \in \{0,1\}^q$, $z_{v,\sigma_2\cdot\sigma_4}' = \sum_{\sigma_1 \in \{0,1\}^p\sigma_3 \in \{0,1\}^q} z_{e,\sigma_1\cdot\sigma_3\sigma_2\cdot\sigma_4}$. Proof is similar to the previous part.
\item[Constraint 5:] For directed edge $e$, $x_{e}' - \sum_{\sigma_1,\sigma_3 \in \{0,1\}^p,\sigma_2\sigma_4\in \{0,1\}^q: \sigma_1 \cdot \sigma_3 \not \leq \sigma_2 \cdot \sigma_4} z_{e,\sigma_1\cdot\sigma_3\sigma_2\sigma_4}' = 0$. This is true by definition of $x_e'$. 
\item[Constraint 6:] For undirected edge $e$, $x_{e}' - \sum_{\sigma_1,\sigma_3 \in \{0,1\}^p,\sigma_2\sigma_4\in \{0,1\}^q: \sigma_1 \cdot \sigma_3 \neq \sigma_2 \cdot \sigma_4} z_{e,\sigma_1\cdot\sigma_3\sigma_2\sigma_4}' = 0$. This is true as well from the definition of $x_e'$. 
\end{itemize}
 
\end{itemize}
\end{proof}

\begin{lemma} The cost $(\bx',\bz')$ is upper bounded by the cost of $\bz$. 
\end{lemma}
\begin{proof}
Recall that $\sigma_0 = 1^q$. We consider three cases based on the type of edge $e$
\begin{itemize}
\item $e = e_t = (u,v)$ for constraint $\cC(u,v)$. 
\begin{eqnarray*}
x_{e_t}' &=& \sum_{\sigma_1,\sigma_2 \in \{0,1\}^{p}, \sigma_3,\sigma_4 \in \{0,1\}^q: \sigma_1\cdot\sigma_3 \not \leq \sigma_2\cdot\sigma_4} z'_{e_t,\sigma_1\cdot\sigma_3\sigma_2\cdot\sigma_4}\\
& = & \sum_{\sigma_1,\sigma_2 \in \{0,1\}^{p}: \sigma_1\cdot\sigma_0\not \leq \sigma_2\cdot\sigma_0} z'_{e_t,\sigma_1\cdot\sigma_0 \sigma_2\cdot\sigma_0}\\
& = & \sum_{\sigma_1,\sigma_2 \in \{0,1\}^{p}: \sigma_1\not \leq \sigma_2} z_{t,\sigma_1\sigma_2}
\end{eqnarray*}

\item $e = e_t = (u,v)$ for constraint $NAE_2(u,v)$

\begin{eqnarray*}
x_{e_t} & = &\sum_{\sigma_1,\sigma_2 \in \{0,1\}^{p}, \sigma_3,\sigma_4 \in \{0,1\}^q: \sigma_1\cdot\sigma_3 \neq \sigma_2\cdot\sigma_4} z'_{e_t,\sigma_1\cdot\sigma_3\sigma_2\cdot\sigma_4}\\
& = & \sum_{\sigma_1,\sigma_2 \in \{0,1\}^{p}: \sigma_1\cdot\sigma_0\neq \sigma_2\cdot\sigma_0} z'_{e_t,\sigma_1\cdot\sigma_0 \sigma_2\cdot\sigma_0}\\
& = & \sum_{\sigma_1,\sigma_2 \in \{0,1\}^{p}: \sigma_1\neq \sigma_2} z_{t,\sigma_1\sigma_2}
\end{eqnarray*}

\item $e = e_t = (u,a_i)$ or $(u,b_j)$ for constraint $\psiai(u)$ or $\psibj(u)$. In such a case $z'_{e_t,\sigma_1 \cdot \sigma_3\sigma_2\cdot\sigma_4}$ is non-zero only if $\sigma_1 = \sigma_2, \sigma_3 = \sigma_4 = \sigma_0$. Hence, 
\begin{equation*}
x_{e_t}' = \sum_{\sigma_1,\sigma_2 \in \{0,1\}^{p},\sigma_3,\sigma_4 \in \{0,1\}^q: \sigma_1\cdot \sigma_3 \neq \sigma_2\sigma_4} z'_{e_t,\sigma_1\cdot\sigma_3\sigma_2\cdot\sigma_4} =0
\end{equation*}
\end{itemize}
Combining the above three facts we get the following. First, infinite
any infinite weight edge $e$ in $G$ has $x'_e = 0$. For any
finite weight edge $x'_e$ is the same as the fractional cost paid
by the corresponding finite weight binary constraint in $\Ic$.

\iffalse
\begin{eqnarray*}
\sum_{e \in E_G} w_G(e) x_e' &\leq& \sum_{t \in T_C, \psitc(t) = \cC} \wtc(t) \sum_{\sigma_1, \sigma_2 \in \{0,1\}^p: \sigma_1 \not \leq \sigma_2} z_{t,\sigma_1\sigma_2} + \sum_{t \in T_C, \psitc(t) = \NAE_2} \wtc(t) \sum_{\sigma_1, \sigma_2 \in \{0,1\}^p: \sigma_1 \neq \sigma_2} z_{t,\sigma_1\sigma_2}\\
& = & \sum_{t \in T_C, \psitc(t) = \cC} \sum_{\sigma_1, \sigma_2 \in \{0,1\}^p} z_{t,\sigma_1\sigma_2}\cC(\sigma_1,\sigma_2)+  \sum_{t \in T_C, \psitc(t) = \NAE_2} \wtc(t) \sum_{\sigma_1, \sigma_2 \in \{0,1\}^p} z_{t,\sigma_1\sigma_2}\NAE_2(\sigma_1,\sigma_2)\\
& \leq & \sum_{t\in T_C} \wtc(t) \sum_{\alpha \in L^{|t|}} z'_{t,\alpha}\cdot \psi_t(\alpha)
\end{eqnarray*}
\fi

Hence, cost of $(\bx',\bz')$ is upper bounded by cost of $\bz$. 
\end{proof}

%\begin{lemma}$\bz'$ as defined above satisfies the following properties:
%\begin{itemize}
%\item $\sum_{\sigma \in \{0,1\}^{p+q}} z_{v,\sigma} = 1$
%\item For $i \in [1,p]$, let $\sigma_1 \in \{0,1\}^p$ such that $J_{\sigma_1} = Y_i$ and $\sigma_2 = 1^q$. Then, $z_{a_i,\sigma_1\cdot\sigma_2}=1$
%\item For $j \in [1,q]$, let $\sigma_1 \in \{0,1\}^p$ such that $J_{\sigma_1} = Z_j$ and $\sigma_2 = 1^q$. Then, $z_{b_j,\sigma_1\cdot\sigma_2}=1$
%\item 
%\end{itemize}
%
%\end{lemma}

\subsubsection{Reduction from \DMulCH to \MinbCSP}
Let the \DMulCH instance be $I_M = (G=(V_G,E_G,w_G:E_G\rightarrow
R^+), (S \cup T, E_H))$. Recall that the supply graph satisfies
assumptions I, II, and III.  We reduce it an equivalent \MinbCSP
instance $I_C = (V_C, T_C, \psitc:T_C\rightarrow \beta_H ,
\wtc:T_C\rightarrow R^+)$ as follows.
\begin{itemize}
\item Vertex Set $V_C = V_G$.
\item $T_C, \psitc, \wtc$ are defined as follows:
\begin{itemize}
\item For every $a_i\in S$, add a tuple $t=(a_i)$ in $T_C$ with $\psitc(t) = \psiai$ and $\wtc(t) = 1$.
\item For every $b_j \in T$, add a tuple $t = (b_j)$ in $T_C$ with $\psitc(t) = \psibj$ and $\wtc(t) = 1$.
\item For every directed edge $e=(u,v) \in E_G$, add a tuple $t = (u,v)$ in $T_C$ with $\psitc(t) = \cC$ and $\wtc(t) = w_G(e)$.
\item For every undirected edge $e = uv \in E_G$, add a tuple $t = (u,v)$ in $T_C$ with $\psitc(t) = \NAE_2$ and $\wtc(t)= w_G(e)$.
\end{itemize}
\end{itemize}

The proof of equivalence between \LabelLP for $I_M$ and \BasicLP for $I_C$
is essentially identical to the proof for the reduction in the other
direction and hence we omit it.

%\vcomment{Should $\psi$ be replaced with $\Psi$ and $w_{T_C}$ with $W_{T_C}$?}

This finishes the proof of Theorem~\ref{thm:label_basic_equivalence}.

\section{Hardness for Non-bipartite Demand graphs}\label{sec:appendix-nonbipartite-hardness}

Here we prove Theorem~\ref{thm:hardness-nonbipartite} on the hardness
of approximation of \DMulCH when $H$ is fixed and may not be bipartite.
Let $\gamma_H$ denote the hardness of approximation for 
\DMulCH. Recall that $\alpha_H$ is the worst-case flow-cut gap for 
\DMulCH.

Let the demand graph be $H$ with $2^p$ vertices, $V_H = \{s_\sigma
\mid \sigma \in \{0,1\}^p\}$. If number of vertices not a power of
two, then we can add dummy isolated vertices without changing the
problem. We find $r=2p$ subgraphs $H_1,\dots,H_r$ such that $H =
H_1 \cup \dots \cup H_r$ and
\begin{itemize}
\item Each $H_i$ is a directed bipartite graph.
\item $\alpha_H \le \sum_{i=1}^r \alpha_{H_i}$. 
\item  For $1 \le i \le r$, there is an approximation preserving reduction from
  \DMulCHi to \DMulCH. Hence, $\gamma_H \geq \gamma_{H_i}$.
\end{itemize}

Since, $H_i$ is bipartite, Theorem~\ref{thm:hardness} implies, under
UGC, that $\gamma_{H_i} \geq \alpha_{H_i}-\eps$. 
Since, $\gamma_H \geq \gamma_{H_i}$
for all $i \in [1,r]$, we have $\gamma_H \geq \frac{1}{r} \sum_{i=1}^r
\gamma_{H_i}$. Therefore,
\begin{equation*}
\gamma_H  \geq \frac{1}{r}\sum_{i=1}^r \gamma_{H_i} \geq \frac{1}{r} \sum_{i=1}^r (\alpha_{H_i} - \eps) \geq \frac{1}{r} \alpha_H - \eps.
\end{equation*}

Since $r = 2\lceil \log k \rceil$ where $k = |V_H|$, we obtain
the proof of Theorem~\ref{thm:hardness-nonbipartite}.

Next, we show how to construct $H_i$ which satisfy the properties
above. For each number $j \in [1,p]$, define $A_j = \{s_\sigma \mid
\sigma \in \{0,1\}^p, \sigma(j) = 0\}, B_j = \{s_\sigma \mid \sigma
\in \{0,1\}^p, \sigma(j) = 1\}$. Let $H_{2j-1}$ be the subgraph of $H$
with vertex set $V_H$ and edge set containing edges of $H$ with head
in $B_j$ and tail in $A_j$. $H_{2j}$ be the subgraph of $H$ with
vertex set $V_H$ and edge set containing edges of $H$ with head in
$A_j$ and tail in $B_j$.
\begin{eqnarray*}
V_{H_{2j-1}}&=& V_{H_{2j}}  =  V_H\\
E_{H_{2j-1}} & = & \{(s_{\sigma_1},s_{\sigma_2}) \in E_H \mid s_{\sigma_1} \in A_j, s_{\sigma_2} \in B_j\}\\
E_{H_{2j}} & = & \{(s_{\sigma_1},s_{\sigma_2}) \in E_H \mid s_{\sigma_1} \in B_j, s_{\sigma_2} \in A_j\}\\
\end{eqnarray*}

By construction, it is clear that $H_{2j-1}, H_{2j}$ are bi-partite. 
\begin{lemma}
$H_i$ as defined above satisfy the following properties:
\begin{itemize}
\item $E_H = \cup_{i=1}^r E_{H_i}$.
\item $\alpha_H \leq \sum_{i=1}^r \alpha_{H_i}$.
\item For $i \in [1,r], \gamma_H \geq \gamma_{H_i}$.
\end{itemize}
\end{lemma}
\begin{proof}
\begin{itemize}
\item Let $e=(s_{\sigma_1},s_{\sigma_2}) \in E_H$. Since, there are no self-loops in $H$, there exists $j \in [1,p]$ such that either $\sigma_1[j] = 1, \sigma_2[j] = 0$ or $\sigma_1[j] = 0, \sigma_2[j] = 1$. In the first case, $e \in E_{H_{2j-1}}$ and in the second case $e \in E_{H_{2j}}$.
\item Given a \DMulCH instance $(G,H)$, idea is to solve $(G,H_i)$ for
  $i \in [1,p]$. Let $I = (G,H)$ be a \DMulCH instance.  Let $\bx$ be
  the optimal solution to \DistLP on $I$. Let $I_i = (G,H_i)$ be the
  instance with the same supply graph $G$ but demand graph $H_i$.
  It is easy to see that $\bx$ is a feasible fractional solution to
  $I_i$ since $H_i$ is a subgraph of $H$. Since the worst-case
  integrality gap for \DMulCHi is $\alpha_{H_i}$, there is a set $E'_i
  \subseteq E_G$ such that $w(E'_i) \le \alpha_{H_i} w(\bx)$ and
  $G-E'_i$ disconnects all demand pairs in $H_i$. Clearly $\cup_i
  E'_i$ is a feasible integral solution to $(G,H)$ since $H = \cup_i
  H_i$. The cost of $\cup_i E'_i$ is at most $(\sum_i \alpha_{H_i})
  w(\bx)$. Since $(G,H)$ was an arbitrary instance of \DMulCH, this
  proves that $\alpha_H \le \sum_i \alpha_{H_i}$.

\item We prove that there is an approximation preserving reduction
  from \DMulCHi to \DMulCH which in turn proves that $\gamma_H \geq
  \gamma_{H_i}$. Assume that $i = 2j-1$ (case when $i = 2j$ is
  similar). Let $(G,H_i)$ be a \DMulCHi instance. $G'$ is defined as
  follows:
\begin{itemize}
\item $V_{G'} = V_G \cup A_j' \cup B_j'$ where $A_j'= \{s_{\sigma}' \mid s_{\sigma} \in A_j\},B_j'= \{s_{\sigma}' \mid s_{\sigma} \in B_j\}$. 
\item $G'$ contains all the edges of $G$ and an infinite edge from $s_{\sigma}'$ to $s_{\sigma}$ for every $s_{\sigma} \in A_j$ and infinite weight edge from $s_{\sigma}$ to $s_{\sigma}'$ for every $s_{\sigma} \in B_j$. 
\end{itemize}

Let $H'$ be a demand graph with vertex $s_\sigma$ renamed as
$s_\sigma'$. Then, $(G',H')$ is a \DMulCH instance. Note that for $s_{\sigma} \in A_j$, $s_{\sigma}'$ in $G'$ has no incoming edge and for $s_{\sigma} \in B_j$, $s_{\sigma}'$ in $G'$ has no outgoing edge. Hence, for \DMulC instance $(G',H')$, we only need to seperate $(s_{\sigma_1}',s_{\sigma_2}')$ if $s_{\sigma_1} \in A_j, s_{\sigma_2} \in B_j$. Hence, \DMulC instances $(G,H_i)$ and $(G',H')$ are equivalent.
\end{itemize}
\end{proof}

\section{Approximating \DMulC with restricted Demand graphs}\label{sec:appendix_dirapproximation}
In this section we prove the following restated theorem.
\dirapproximation*

Let $G=(V,E)$ and $H=(V,F)$ be the supply and demand graph for a given
instance of \DMulC.  We prove this theorem by providing a generic
randomized rounding algorithm that given a fractional solution $\bx$
to LP~\ref{fig:mulc_lp} for an instance $(G,H)$ of \DMulC returns a
feasible solution. This algorithm is inspired by our recent rounding
scheme for \DMCut \cite{ChekuriM16}.  Let $\alpha_e x_e$ be the
probability that a given edge $e$ in the supply graph $G$ is cut by
the algorithm. If $\alpha_e \le (k-1)$ for all $e \in E$ then clearly
the expected cost of the feasible solution is at most $(k-1)$ and we
are done. However, if there is some edge $e$ such that $\alpha_e >
(k-1)$ we show that $H$ contains an induced $k$-matching-extension.

Let $\bx$ be a feasible solution to LP~\ref{fig:mulc_lp}. For $u,v \in
V$, define $d(u,v)$ to be the shortest path length in $G$ from vertex
$u$ to vertex $v$ using lengths $x_e$. We also define another
parameter $d_1(u,v)$ for each pair of vertices $u,v \in V$. $d_1(u,v)$
is the minimum non-negative number such that if we add an edge $uv$ in
$G$ with $x_{uv} = d_1(u,v)$ then $u$ is still seperated from all the
vertices it has to be seperated from. Formally, for $u,v \in V$,
$d_1(u,v) := \max(0,1-\min_{v' \in V, uv'\in F} d(v,v'))$.  If for
some vertex $u$, there is no demand edge leaving $u$ in $F$ then we
define $d_1(u,v) = 0$ for all $v \in V$. The following properties of
$d_1$ are easy to verify.

\begin{lemma}\label{lem:d1_property}
$d_1(u,v)$ satisfies the following properties:
\begin{itemize}
\item $\forall u\in V, d_1(u,u)=0$
\item $\forall (u,v) \in F, v' \in V$, $d_1(u,v') + d(v',v) \geq 1$
\item If $d_1(u,v) \neq 0$, then there exists $(u,v') \in F$ such that 
$d_1(u,v)+d(v,v')=1$
\item $\forall u \in V, (a,b) \in E, d_1(u,b) - d_1(u,a) \leq x_{ab}$
\end{itemize}
\end{lemma}
\iffalse
\begin{proof}
  First three properties can be proved by simply writing the
  definition of $d_1(u,v)$. To prove $d_1(u,b) -d_1(u,a) \leq x_{ab}$,
  consider $v', v''$ such that $uv', uv'' \in E_H$ and $d(a,v')+
  d_1(u,a) = d(b,v'') + d_1(u,b) = 1$.  From definition of $d_1(u,a)$
  and $d_1(u,b)$ we have $d(a,v') \leq d(a,v'')$ and $d(b,v'') \leq
  d(a,v'')$.  \vspace{-2mm}
\begin{align*}
  d_1(u,b) - d_1(u,a) & = d(a,v') - d(b,v'') \\
  & \leq d(a,v'') - d(b,v'')\\
  & \leq x_{ab}
\end{align*}
Last inequality holds as $d(a,v'')$ is the shortest path distance from $a$ to $v''$ and hence, must satisfy the triangle inequality.
\end{proof}
\fi

Next, we do a simple ball cut rounding around all the vertices as per
$d_1(u,v)$. We pick a number $\theta \in (0,1)$ uniformly at
random. For all $u \in V$, we consider $\theta$ radius ball around $u$
for all $u \in V$; $B_u = \{v \in V \mid d_1(u,v) \leq \theta\}$. And
then cut all the edges leaving the set $B_u$; $\delta^+(B_u) = \{vv'
\in E_G \mid v \in B_u, v' \not \in B_u\}$. Note that it is crucial
that the same $\theta$ is used for all $u$. A formal description
of the algorithm is given in Algorithm~\ref{alg:directed_cut_rounding_scheme}.

\begin{algorithm}
	\caption{Rounding for \DMulC}
	\label{alg:directed_cut_rounding_scheme}
	\begin{algorithmic}[1]
		\STATE Given a feasible solution $\bx$ to LP~\ref{fig:mulc_lp}
		\STATE For all $u,v \in V$, compute $d(u,v)$= shortest path length from $u$ to $v$ according to lengths $x_e$
		\STATE For all $u,v \in V$, compute $d_1(u,v) =\max(0,1-\min_{v' \in V, uv'\in E_H} d(v,v'))$ 
		\STATE Pick $\theta \in (0,1)$ uniformly at random
		\STATE $B_u = \{v \in V \mid d_1(u,v) \leq \theta\}$
                \STATE $E' = \cup_{u \in V} \delta^+(B_u)$
		\STATE Return $E'$
	\end{algorithmic}
\end{algorithm}

It is easy to argue that $E'$ returned by the algorithm is a feasible
\DMulC solution. By Lemma~\ref{lem:d1_property} for all $uv \in F$,
$d_1(u,v) \geq 1$ and since $\theta <1$, we have $u \in B_u, v \not
\in B_u$. We remove all the edges going out of the set $B_u$ and
hence, cut all the paths from $u$ to $v$. We only need to prove that
probability of an edge $e$ being cut by the algorithm is at most
$(k-1)x_e$. To prove that, we need the following lemma which shows
that for any vertex $v$, number of $u_i$ with different values of
$d_1(u_i,v)$ is at most $k-1$.

\begin{lemma}\label{lem:different_distances}
  If for some $v \in V$ there exists $u_1,\dots,u_k$ such that $0\neq
  d_1(u_i,v) \neq d_1(u_j,v)$ for all $i \neq j$, then the demand
  graph $H$ contains an induced $k$-matching extension.
\end{lemma}
\begin{proof}
  Rename the vertices $u_1,\dots,u_k$ such that $d_1(u_1,v) >\dots >
  d(u_k,v)>0$. By Lemma~\ref{lem:d1_property}, there exists
  $v_1',\dots,v_k'$ such that $u_i v_i' \in F$ and $d_1(u_i,v) = 1 -
  d(v,v_i')$. Consider the subgraph of $H$ induced by the vertices
  $s_1,\dots,s_k,t_1,\dots,t_k$ where $s_i = u_i, t_i = v_i'$. Edge
  $s_it_i\in F$ as $u_iv_i' \in F$. By construction $s_1,\ldots,s_k$
  are distinct. We also argue that $t_1,\ldots,t_k$ are distinct.
  Suppose $t_i = t_j$ for $i < j$. We have $d_1(u_i,v'_i) = d_1(u_i,v)
  + d(v,v'_i) = 1$ and $u_iv'_i \in F$. Since $d_1(u_j,v) < d(u_i,v)$
  we have $d_1(u_j,v'_j) = d_1(u_j,v'_i) < 1$, however $u_jv'_j \in F$
  which is a contradiction.   

  For $i>j$, $d_1(s_i,v) + d(v,t_j) = d_1(u_i,v) + 1-d_1(u_j,v) <1$. 
  By lemma~\ref{lem:d1_property}, $s_it_j \not \in F$. Thus we have
  shown that the graph induced on  $s_1,\dots,s_k,t_1,\dots,t_k$ proves
  that $H$ contains an induced $k$-matching extension.
\end{proof}

\begin{proofof}{Theorem~\ref{thm:dir_approximation}} 
  We start by solving LP~\ref{fig:mulc_lp} and then perform the
  rounding scheme as per
  Algorithm~\ref{alg:directed_cut_rounding_scheme}. As argued above,
  for all $uv \in E_H$, $u \in B_u, v \not \in B_u$ and we cut the
  edges going out of $B_u$. Hence, there is no path from $u$ to $v$ in
  $G-E'$ and $E'$ is a feasible \DMulC solution.

  We claim that $\Pr[e \in E'] \leq (k-1)x_e$ for all $e \in
  E_G$. Once we have this property, by linearity of expectation, the
  expected cost of $E'$ can be bounded by $(k-1)$ times the LP cost:
  $\Ex[\sum_{e \in E'} w_e] \leq (k-1) \sum_{e \in E_G} w_e x_e$.

  Now we prove the preceding claim.  Consider an edge $e = (a,b) \in
  E$. Edge $e \in E'$ only if for some $u \in V$, $e \in
  \delta^+(B_u)$ and this holds only if $\theta \in
  [d_1(u,a),d_1(u,b))$. By Lemma~\ref{lem:d1_property}, $d_1(u,b) \leq
  d_1(u,a)+x_e$. Hence, $e \in \delta^+(B_u)$ , if $\theta
  \in [d_1(u,b)-x_{ab}, d_1(u,b))$. Denote this interval by $I_u(e)$.

  By Lemma~\ref{lem:different_distances}, there are at most $k-1$
  distinct elements in the set $\{d_1(u,b) \mid u \in V\}$. This
  implies that there are at most $k-1$ distinct intervals
  $I_u(e)$. In other words there exists $u_1,\dots,u_{r}, r\leq k-1$
  such that $\cup_{u \in V} I_u(e) = \cup_{i=1}^{r} I_{u_i}(e)$.
\begin{align*}
  \Pr[ab \in E'] &\leq \Pr[ \theta \in \cup_{u \in V} I_u(e)]\\
  & =  \Pr[\theta \in \cup_{i=1}^{r} I_{u_i}(e)]\\
  & \leq \sum_{i=1}^r \Pr[\theta \in I_{u_i}(e)]\\
  & \leq r\cdot x_e\\
  & \leq (k-1)x_e.
\end{align*}

Penultimate inequality follows from the fact that $I_{u_i}(e)$ has
length $x_e$ and $\theta$ is chosen uniformly at random from $[0,1)$.
\end{proofof}

\end{document}